\documentclass{article}

\pdfoutput=1

\usepackage[usenames,dvipsnames]{xcolor}
\usepackage{amsmath,amssymb}
\usepackage{amsthm}
\usepackage{natbib}
\usepackage{graphicx}
\usepackage{subfigure}
\usepackage{algorithm}
\usepackage{algorithmic}
\usepackage{multirow}

\newcommand{\Gs}{\mathcal{G}}
\newcommand{\Hs}{\mathcal{H}}
\newcommand{\Vs}{\mathcal{V}}
\newcommand{\Es}{\mathcal{E}}
\newcommand{\inter}{\cap}
\newcommand{\union}{\cup}
\newcommand{\ft}{\tilde{f}}
\newcommand{\R}{\mathbb{R}}
\newcommand{\nlsum}{\sum\nolimits}
\newcommand{\Gf}{\mathsf{G}}
\newcommand{\Vf}{\mathsf{V}}
\newcommand{\vf}{\mathsf{v}}
\newcommand{\Ef}{\mathsf{E}}
\newcommand{\Fs}{\mathcal{F}}

\newcommand{\mcoocut}{{\sc MinCoopCut}}
\newcommand{\mcut}{{\sc Minimum $(s,t)$-Cut}}
\newcommand{\fa}{\hat{f}}   
\newcommand{\Sh}{\widehat{S}}
\newcommand{\Ch}{\widehat{C}}
\newcommand{\Cs}{\mathcal{C}}
\newcommand{\famo}{\hat{f}_{add}}   
\newcommand{\fae}{\hat{f}_{ea}}   
\newcommand{\fap}{\hat{f}_{pf}}  
\newcommand{\fs}{\hat{f}_s}  
\newcommand{\pr}{\mathrm{Pr}}

\newtheorem{prob}{Problem}

\newtheorem{theorem}{Theorem}
\newtheorem{lemma}{Lemma}
\newtheorem{proposition}{Proposition}

\newtheorem{corollary}{Corollary}
\theoremstyle{definition}
\newtheorem{defn}{Definition}

\DeclareMathOperator*{\argmin}{argmin}

\def\restriction#1#2{\mathchoice
              {\setbox1\hbox{${\displaystyle #1}_{\scriptstyle #2}$}
              \restrictionaux{#1}{#2}}
              {\setbox1\hbox{${\textstyle #1}_{\scriptstyle #2}$}
              \restrictionaux{#1}{#2}}
              {\setbox1\hbox{${\scriptstyle #1}_{\scriptscriptstyle #2}$}
              \restrictionaux{#1}{#2}}
              {\setbox1\hbox{${\scriptscriptstyle #1}_{\scriptscriptstyle #2}$}
              \restrictionaux{#1}{#2}}}
\def\restrictionaux#1#2{{#1\,\smash{\vrule height .8\ht1 depth .85\dp1}}_{\,#2}}

\usepackage{ifthen}

\begin{document}
\title{Graph Cuts with Interacting Edge Costs -- Examples, Approximations, and Algorithms}
\date{}
\author{Stefanie Jegelka and Jeff Bilmes}
\maketitle

\begin{abstract}
We study an extension of the classical graph cut problem, wherein we replace the modular (sum of edge weights) cost function by a submodular set function defined over graph edges.
Special cases of this problem have appeared in different applications in signal processing, machine learning, and computer vision. In this paper, we connect these applications via the generic formulation of ``cooperative graph cuts'', for which we study complexity, algorithms, and connections to polymatroidal network flows. Finally, we compare the proposed algorithms empirically.
\end{abstract}

\section{Introduction}
\label{intro}

Graphs have been a ubiquitous modeling tool in areas as diverse as operations research, signal processing and machine learning. Graphical representations reveal structure in the problem that is often the key to obtaining efficient algorithms for
real-world data analysis problems.
As a prominent example, the \mcut{} 
problem underlies important problems in
 low-level computer vision \citep{boyVek06} (e.g.,
image segmentation and regularization), probabilistic inference in
graphical models \citep{gps89,rkat08}, and for representing pseudo-boolean functions in computer vision and constraint satisfaction problems \citep{kb05,rrl11,zcj09}. The reduction to cuts has had a tremendous practical impact.

The algorithmic efficiency of cuts comes at the price of reduced modeling power: graph cuts model problems that correspond to a special class of functions (a sub-class of submodular functions defined on the nodes of the graph \citep{zcj09}). Section~\ref{sec:motiv-spec-cases} lists applications that do not fall into this category. Motivated by these limitations, this paper 
studies a non-additive generalization of \mcut{}, where the cut cost function is a \emph{submodular set function} over graph edges. 

A set function $f: 2^{\Es} \to \mathbb{R}$ defined on subsets of a ground set $\Es$ is submodular if it satisfies \emph{diminishing
  marginal costs}: for all sets $A \subset B \subseteq \Es$ and $e \in
\Es \setminus B$, it holds that
\begin{align}\label{eq:diminreturns}
  f(A \union \{e\}) - f(A) \geq f(B \union \{e\}) - f(B).
\end{align}
This generalization -- we refer to it as \emph{Cooperative Cut} -- introduces dependencies between edges, and expresses a wider set of functions than graph cuts.

\subsection{Minimum Cut and Minimum Cooperative Cut}

The \mcut{} problem is defined as follows.
\begin{prob}[\mcut]
  Given a weighted graph $\Gs = (\Vs,\Es,w)$ with terminal nodes $s, t \in \Vs$,
  find a cut $C \subseteq \Es$ of minimum cost $w(C) = \sum_{e \in
    C}w(e)$. A \emph{cut} is a set of edges whose removal
  disconnects all paths between $s$ and $t$.
\end{prob}
We assume throughout that $w \geq 0$.
Many very efficient algorithms are known to solve \mcut; the reader
is referred to
\citep{amo93,schrijver2004} for an overview.

In graph cuts, the cost of any given cut $C \subseteq \Es$ is
a sum $w(C) = \sum_{e \in C}w(e)$ of edge weights. 
This function is said to be
\emph{modular} or, equivalently, \emph{additive} on the edge set $\Es$. 
It implies two important modeling characteristics for graph cuts: First, the nonnegativity of the weights can only penalize two nodes for being separated --- this introduces a form of positive correlation between nodes, also hence this is sometimes referred to as \emph{attractive} potentials in the computer vision community. Second, modular edge weights express interactions of only two nodes at a time. That is, the additive contribution $w(e)$ to the cost $\sum_{e \in C}w(e)$ of a cut $C$ by a given edge $e \in C$ is the same regardless of the cut in which the edge $e$ is
considered. 

Several applications, however, are more accurately modeled when allowing 
non-additive interactions between edge weights. We survey some
examples and applications in Section~\ref{sec:motiv-spec-cases}.
These examples are captured when replacing the modular cost function
$w$ by a nonnegative and nondecreasing \emph{submodular} set function
over graph edges.  The definition \eqref{eq:diminreturns} implies that
with submodular edge weights, the cost of an additional edge depends
on which other edges are contained in the cut. This non-additivity
allows specific long-range dependencies between multiple (pairs
of) nodes simultaneously that cannot be expressed by graph cuts.
These observations motivate the definition of \emph{cooperative graph
  cuts}.

\begin{prob}[Minimum cooperative cut (\mcoocut)]\label{prob:cc}
  Given a graph $\Gs = (\Vs,\Es,f)$ with terminal nodes $s, t \in \Vs$
  and a nonnegative, monotone nondecreasing submodular function $f: 2^\Es \to
  \mathbb{R}_+$ defined on subsets of edges, find an $(s,t)$-cut $C \subseteq \Es$ of minimum cost $f(C)$.
\end{prob}
A set function $f$ is \emph{nondecreasing} 
if $A \subseteq B \subseteq \Es$ implies that $f(A) \leq f(B)$.
\mcoocut{} is a constrained submodular minimization problem:
\begin{alignat}{2}
  \label{eq:cc}
  \text{minimize } f(C) \quad \text{ subject to } C \subseteq \Es \text{ is an
    $(s,t)$-cut in } \Gs.
\end{alignat}
As cooperative cuts employ the same graph structures as standard graph
cuts, they easily integrate into and extend many of the applications of graph cuts.
We note, however, that the graph $\Gs$ need not have any relationship
to the submodular function $f$ other than the fact that the edges
of the graph $\Gs$ constitute the ground set of $f$.

\subsection{Relation to the literature}
Cooperative graph cuts relate to the recent literature in two aspects.
First, a number of models in signal processing, computer vision, security, and machine learning 
are special cases of cooperative cuts, as is discussed in Section~\ref{sec:motiv-spec-cases}.

Second, recent interest has emerged in the 
literature regarding the implications of extending
classical combinatorial problems (such as {\sc Shortest Path}, {\sc Minimum
Spanning Tree}, or {\sc Set Cover}) from a sum-of-weights to submodular cost
functions
\citep{sf08,in09b,gktp09,gtw10,jb11,jb11inference,ky09,hms07,zctz09,
  bauBerBu13}. None of this work, however, has addressed cuts. 
In this work, we provide lower and upper bounds on the approximation factor of \mcoocut{}. 

One approach to solving \mcoocut{} is via relaxations.
For \mcut{}, a celebrated result of \citep{ff56,df54} states that the dual of the linear programming relaxation is a {\sc Maximum Flow} problem, and that their optimal values coincide. 
We refer to the ratio between the maximum flow value (i.e., the optimal value of the relaxation), and the optimal value of the discrete cut, as the \emph{flow-cut gap}. For \mcut{}, this ratio is one. 
In Section~\ref{sec:relax}, we formulate a convex relaxation of \mcoocut{} whose dual problem is a generalized flow problem, where submodular capacity constraints are
 placed not only on individual edges but on arbitrary sets of edges
 simultaneously. The flow-cut gap for this problem can be on the order of $n$, the number of nodes. In contrast, the related polymatroidal maximum flow problem \citep{lm82,h82} (defined in Section~\ref{subsec:pmf}) still has a flow-cut gap of one. Polymatroidal flows are equivalent to submodular flows, and have recently gained attention for modeling information flow in
 wireless networks \citep{krv11,kv11,ckrv12}. Their dual problem is a minimum cut problem where the edge weights are defined by a convolution of local submodular functions \citep{lo83}. Such convolutions are generally not submodular (see Equation~\eqref{eq:pmf}).

\subsection{Summary of main contributions}

In this paper, we survey diverse examples of cooperative cuts in different applications, and
provide a detailed theoretical analysis: 
\begin{itemize}
\item We show a lower bound of $\Omega(\sqrt{n})$ on the approximation factor for
  the general \mcoocut{} problem. 
\item We analyze two families of approximation
  algorithms. The first relies on substituting the submodular cost
  function by a tractable approximation. The second family consists of
  rounding algorithms that build on the relaxation of the
  problem.
  Both families contain algorithms that use
  a partitioning of edges into node incidence
  sets, but in different ways. 
\item   We provide a lower bound of $n-1$ on the flow-cut gap, and relate
  it to different families of submodular functions. 
\end{itemize}
In Section~\ref{subsec:pmf} we draw connections to polymatroidal flows \citep{lm82,h82}. The non-additive cut cost function used in the resulting approximation is solvable exactly and may itself be interesting to consider as an exactly solvable class of e.g.\ higher-order potentials in computer vision.
The paper concludes with a discussion of open problems. In particular, the results of this paper motivate a wider and more refined study of the complexity and expressive power of non-additive graph cut problems.

The paper is structured as follows. In Section~\ref{sec:motiv-spec-cases} we discuss various instances of cooperative cuts and their properties. The complexity of \mcoocut{} is addressed in Section~\ref{sec:complexity}, convex relaxations in Section~\ref{sec:relax},
and algorithmic approaches in Section~\ref{sec:algo}.

\subsection{Notation and technical preliminaries} 
\label{sec:prel-notat}

Throughout this paper, we are given a directed
graph\footnote{Undirected graphs can be reduced to bidirected graphs.}
$\Gs = (\Vs, \Es)$ with $n$ nodes and $m$ edges, and terminal nodes
$s,t \in \Vs$. The function
$f:2^\Es \to \mathbb{R}_+$ is submodular and monotone nondecreasing, where by $2^\Es$ we denote the power set
of $\Es$.  We assume $f$ to be \emph{normalized}, i.e., $f(\emptyset)
= 0$. 
Equivalently to Definition~\eqref{eq:diminreturns}, the function $f$ is submodular if
for all $A, B \subseteq \Es$, it
holds that
\begin{equation}
  f(A) + f(B) \geq f(A \union B) + f(A \inter B). \label{eq:subdef}
\end{equation}
The function $f$ generalizes commonly used modular edge weight functions
$w: \Es \to \mathbb{R}_+$ that
satisfy Equation~\ref{eq:subdef} with equality.
We denote the \emph{marginal cost} of an element $e \in \Es$ with respect to a set $A \subseteq \Es$ by
$f(e \mid A) \triangleq f(A \union \{e\}) - f(A)$.
A function $f(A) = g(w(A))$ for a nonnegative modular function $w$ and a concave scalar function $g$ is always submodular.

For any node $v \in \Vs$, let $\delta^+(v) = \{(v,u) \in \Es\}$ be the
set of edges directed out of $v$, and $\delta^-(v) = \{(u,v) \in \Es\}$ be
the set of edges into $v$. Together, these two directed sets form the (undirected) incidence set
$\delta(v) = \delta^+(v) \union \delta^-(v)$.
These definitions directly extend to sets of nodes:
for a set $S \subseteq \Vs$ of nodes, $\delta^+(S) = \{(v,u) \in \Es : v \in S, u \notin S \}$ is the
set of edges leaving $S$. Without loss of generality, we assume all
graphs are simple.

The \emph{Lov\'asz extension} $\ft:[0,1]^m \to \R$ of the submodular
function $f$ is its lower convex envelope and is defined as
follows~\citep{lo83}. Given a vector $x \in [0,1]^m$, we can uniquely
decompose $x$ into its level sets $\{ B_j \}_j$ as $x = \sum_j
\lambda_j \chi_{B_j}$ where $B_1 \subset B_2 \subset \dots$ are
distinct subsets.  Here and in the following, $\chi_B \in [0,1]^m$ is the
characteristic vector of the set $B$, with $\chi_B(e) = 1$ if $e \in
B$, and $\chi_B(e)=0$ otherwise. Then $\ft(x) =
\sum_j\lambda_jf(B_j)$. This construction illustrates that
$\ft(\chi_{B}) = f(B)$ for any set $B$.  The Lov\'{a}sz extension can
be computed by sorting the entries of the argument $x$ in $O(m\log m)$
time and calling $f$ $m$ times.

A \emph{separator} of a submodular function $f$ is a set $S \subseteq
\Es$ with the property that $f(S) + f(\Es \setminus S) = f(\Es)$,
implying that for any $B \subseteq \Es$, $f(B) = f(B \cap S) + f(B
\cap (\Es \setminus S))$.  If $S$ is a minimal separator, then we say
that $f$ \emph{couples} all edges in $S$. For the edges within a
minimal separator, $f$ is strictly subadditive: $\sum_{e \in S}f(e) >
f(S)$ for $|S|>1$. That means, the joint cost of this set of edges is smaller than the sum of their individual costs.

\subsection{Node functions induced by submodular edge weights}
\begin{figure}
  \begin{minipage}{0.5\linewidth}
    \centering
    \includegraphics[width=0.55\textwidth,height=0.5\textwidth]{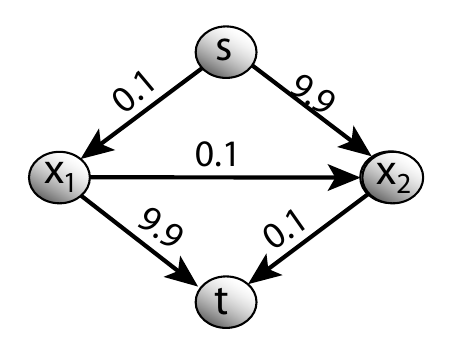}
  \end{minipage}
  \begin{minipage}{0.5\linewidth}
    Let $f(A) = \sqrt{\sum_{e \in A}w(e)}$, so\\
    $h(X) = \sqrt{\sum_{e \in \delta^+(X)}w(e)}$. \\ Then
    $h$ is not submodular:    
  \end{minipage}
  \begin{align*}
    h(\{s,x_1\}) + h(\{s, x_2\}) = \sqrt{19.9} +
      \sqrt{0.2}
      \;\;\;<\;\;\; 2\sqrt{10} =
    h(\{s\}) + h(\{s,x_1,x_2\})
  \end{align*}
  \caption{The node function induced by a cooperative cut is in
    general not submodular. The above $h$ violates
    inequality~\eqref{eq:subdef} for $A = \{s,x_1\}$, $B = \{s, x_2\}$
    but satisfies it (strictly) for 
    $A = \{t,x_1\}$, $B = \{t, x_2\}$.
   }
  \label{fig:notsubmod}
\end{figure}
Every cost function $f$ on graph edges defines a cut function $h: 2^\Vs \to \mathbb{R}_+$ on sets $X \subseteq \Vs$ of nodes:
\begin{align}\label{eq:nodefunc}
  h(X) \triangleq f(\delta^+(X)).
\end{align}
It is well known that if $f$ is a (modular) sum of nonnegative edge
weights, then $h$ is submodular \citep{schrijver2004}. In fact, the following
is true:
\begin{proposition}
The function $f$ is a non-negative modular function on the edge
set if and only if $h(X) = f(\delta^+(X))$ is a submodular
function on the nodes for all edge-induced sub-graphs of $\Gs = (\Vs, \Vs \times \Vs)$.
\label{prop:submdoular_iff_modular}
\end{proposition}
If, however,
$f$ is an arbitrary nondecreasing \emph{submodular} function, then
this is not always the case, as Figure~\ref{fig:notsubmod}
illustrates. Proposition~\ref{prop:properties_h},
proven in Appendix~\ref{app:properties},
summarizes some key properties of $h$.
\begin{proposition}
\label{prop:properties_h}
  Let $h: 2^\Vs \to \R$, $h(X) \triangleq f(\delta^+(X))$ be the node function induced by a cooperative cut with nonnegative nondecreasing submodular cost function $f$. Then:
  \begin{enumerate}
  \item $h$ is not always submodular, and
  \item $h$ is subadditive, i.e., $h(A) + h(B) \geq h(A \cup B)$ for any $A, B \subseteq \Vs$.
  \end{enumerate}
\end{proposition}
The non-submodularity of $h$ shows that cooperative cuts are strictly more general than (modular-weight) graph cuts.
In some cases, the function $h$ is submodular. One obvious sufficient condition is that $f$ is nonnegative and modular, but this condition is not necessary as shown in the following.
\begin{proposition}
\label{prop:symm}
  Let $f$ be monotone submodular and permutation symmetric in the sense that $f(A) = f(\sigma(A))$ for any permutation $\sigma$ of the set $\Es$. If $\Gs$ is a complete graph, then $h$ is submodular.
\end{proposition}
\begin{proof} 
  Symmetry implies that $f$ is of the form $f(A) = g(|A|)$ for a scalar function $g$. Submodularity of $f$ implies that there is always a function $g'$ that interpolates $g$ on $\R\setminus \{0, 1, \ldots,m\}$, i.e., $f(A) = g'(|A|) = g(|A|)$, and $g'$ is a piecewise linear concave function. Let $E_{XY}$ be the edges between sets $X$ and $Y$. 
The submodularity of $h$ is identical to the condition that for all $X \subseteq \Vs$, $x,y \notin X$, it holds that
  \begin{align}
    h(X \union x) + h(X \union y) \geq h(X) + h(X \union x \union y).
  \end{align}
  Let $R = \Vs \setminus (X \union x \union y)$. By the concavity and
  monotonicity of $g'$ we have
  \begin{align*}
    &h(X) + h(X \union x \union y)\\
    &\quad = g'(|E_{XR}|+|E_{Xx}| + |E_{Xy}|) + g'(|E_{XR}|+|E_{Rx}| + |E_{Ry}|)\\
    &\quad = g'(|X||R| + 2|X|) + g'(|X||R| + 2|R|)\\
    &\quad \leq 2 g'(|X||R| + |X| + |R|)\\
    &\quad \leq g'(|E_{XR}| + |E_{Xy}| + |E_{Rx}| + |E_{xy}|) + g'(|E_{XR}| + |E_{Xx}| + |E_{Ry}| + |E_{xy}|)\\
    &\quad= h(X \union x) + h(X \union y),
  \end{align*}
  and hence submodularity of $h$ follows. \qed
\end{proof}
Note that if $\Gs$ is not complete, then $h$ might no longer be submodular. An exact (possibly algebraic or graph-theoretic) characterization of the conditions on $\Gs$ and $f$ that imply submodularity of $h$ is currently an open problem.

\section{Motivation and special cases}
\label{sec:motiv-spec-cases}

We begin by surveying special cases of cooperative cuts from
applications in signal processing, machine learning, and computer
vision.  Notably, some of these applications lead to submodular node
functions $h$ as defined in \eqref{eq:nodefunc} and are hence
polynomial-time solvable, while for others $h$ is not submodular.  We
first discuss the latter case which motivated this paper, and then
special submodular cases. Additional special cases are discussed in
Section~\ref{sec:easier}.

\subsection{General, non-submodular examples}
\label{subsec:nonsubmod-ex}

\paragraph{Image segmentation.} The classical task of segmenting an
image into a foreground object and its background is commonly
formulated as a \emph{maximum a posteriori} (MAP) inference problem in
a Markov Random Field (MRF) or Conditional Random Field (CRF). 
If the potential functions of the random field are submodular functions (of the node variables), then the MAP solution can be computed efficiently via the minimum cut in an auxiliary graph~\citep{gps89,bj01,kolmogorov2004energy}.

While these graph cut models have been very successful in computer vision,
they suffer from known shortcomings.
For example, the cut model implicitly penalizes the length of the object boundary,
or equivalently the length of a corresponding graph cut around the
object.  As a result, MAP solutions (minimum cuts) tend to
truncate fine parts of an object (such as branches of trees, or animal
hair), and to neglect carving out holes (such as the mesh grid of a
fan, or written letters on paper). This tendency is aggravated if the image
has regions of low contrast, where local information is insufficient
to determine correct object boundaries.

A solution to both of these problems is proposed in \citep{jb11}. It relies on the continuation of ``obvious'' object boundaries: one aims to reduce the cut cost if the cut consists of edges (pairs of pixels) with \emph{similar appearance}. This aim is impossible to model with a modular-cost graph cut where edges are independent. Hence, \cite{jb11} replace the graph cut by a cooperative cut that lowers the cost for sets of similar edges. Specifically, the edges in the image graph are partitioned into groups $S_i$ of similar edges, where similarity is defined via the adjacent pixels (nodes), and the cut cost is
\begin{equation}\label{eq:fgroup}
  f(C) = \sum_{i=1}^k g_i(w(C \inter S_i)),
\end{equation}
where the $g_i$ are increasing, strictly concave functions, and $w(C)
= \sum_{e \in C}w(e)$ is a sum of nonnegative weights. 

From the viewpoint of graphical models, the function $h$ induced by \eqref{eq:fgroup} is a \emph{higher-order} potential, i.e., a polynomial of degree much larger than 2.
The model \eqref{eq:fgroup} also applies to multi-label (scene labeling) problems and other computer vision tasks \citep{kohliOJ13,silberman14,taniai15}.

An alternative cooperative cut function 
has been studied to improve image segmentation results:
\begin{equation}
  \label{eq:max}
  f(C) = \max_{e \in C}w(e).
\end{equation}
Contrary to the cost function~(\ref{eq:fgroup}),
the function~(\ref{eq:max}) couples all edges in the
grid graph uniformly, without any similarity constraints or grouping.
As a result, the cost of \emph{any} long cut is discounted.
\citet{sg07} and \citet{aac09} derive this
function as the $\ell_{\infty}$ norm of the (gradient) vector of pixel
differences; this vector is the edge indicator vector $y$ in the
relaxation we define in Section~\ref{sec:relax}.
Conversely, the relaxation of the cooperative cut problem leads
to new, possibly non-uniform and group-structured regularization terms
\citep{mythesis}.

\paragraph{Label Cuts.} In computer security, \emph{attack graphs} are state graphs modeling the steps of an
intrusion. Each
transition edge is labeled by an atomic action $a$, and blocking
an action $a$ blocks the set of all associated edges $S_a \subseteq \Es$ that carry label $a$. To prevent
an intrusion, one must separate the initial state $s$ from the goal
state $t$ by blocking (cutting) appropriate edges. The cost of cutting a set of edges is
the cost of blocking the associated actions (labels), and paying for one
action $a$ accounts for all edges in $S_a$.
If each action has a cost $c(a)$, then a \emph{minimum label cut} that minimizes the submodular cost function
\begin{equation}
  f(C) = \nlsum_a c(a)\min\{1, |C \inter S_a|\}
\end{equation}
indicates the lowest-cost prevention of an intrusion \citep{jsw02}.

\paragraph{Sparse separators of Dynamic Bayesian Networks.} 
 A graphical model $\Gf
= (\Vf,\Ef)$ defines a family of probability
distributions. It has a node $\vf_i$ for each random variable $x_i$,
and any represented distribution $p(x)$ must factor with respect to
the edges of the graph as $p(x) \propto
\prod_{(\vf_i,\vf_j)\in\Ef}\psi_{ij}(x_i,x_j)$. A \emph{dynamic
  graphical model} (DGM) \citep{bi10} consists of three template parts: a prologue $\Gf^p
= (\Vf^p,\Ef^p)$, a chunk $\Gf^c = (\Vf^c,\Ef^c)$ and an epilogue
$\Gf^e = (\Vf^e,\Ef^e)$. Given a length $\tau$, an \emph{unrolling} of
the template is a model that begins with $\Gf^p$ on the left, followed
by $\tau+1$ repetitions of the ``chunk'' part $\Gf^c$ and ending in
the epilogue $\Gf^e$. 

To perform inference efficiently, a periodic section of the partially
unrolled model is identified on which an effective inference strategy
(e.g., a graph triangulation, an elimination order, or an
approximate inference method) is developed and then repeatedly used for the
complete duration of the model unrolled to any length. This periodic
section has boundaries corresponding to separators in the original
model \citep{bi10} which are called the interface separators.
Importantly, the efficiency of any inference algorithm derived within
the periodic section depends critically on properties of the
interface, since the variables within must become a clique.

In general, the computational cost of inference is lower bounded, and
the memory cost of inference is exactly given, by the size of the
joint state space of the interface variables.  A ``small'' separator
corresponds to a minimum vertex cut in the graphical model, where the
cost function measures the size of the joint state space. Vertex cuts
can be rephrased as standard edge cuts. Often, a modular cost function
suffices for good results. Sometimes, however, a more general cost
function is needed. In \citet{bilmes2003-tridbn}, for example (which
is our original motivation for studying \mcoocut),
a state space function that considers deterministic relationships
between variables is found to significantly decrease
inference costs.

An example of a function that respects determinism is the
following. In a Bayesian network possessing determinism, let $D$ be the
subset of fully deterministic nodes.  That means any $x_i \in D$ is a
deterministic function of the variables corresponding to its parent
nodes $\mathrm{par}(i)$ meaning $p(x_i|x_{\mathrm{par}(i)}) =
\mathbf{1}[x_i = g(x_{\mathrm{par}(i)})]$ for some deterministic
function $g$. Let $\mathcal{D}_i$ be the state space of variable
$x_i$. Furthermore, given a set $A$ of variables, let $A_D = \{ x_i
\in A \inter D \mid \mathrm{par}(i) \subseteq A\}$ be its subset of
fully determined variables. If the state space of a deterministic
variable is not restricted by fixing a subset of its parents, then the
function measuring the state space of a set of variables $A$ is $f(A)
= \prod_{x_i \in A\setminus A_D}|\mathcal{D}_i|$. The logarithm of
this function is a submodular function, and therefore the problem of
finding a good separator is a cooperative (vertex) cut problem. In fact, this
function is a lower bound on the computational complexity of
inference, and corresponds exactly to the memory complexity since
memory need be retained only at the boundaries between repeated
sections in a DGM.

More generally, a similar slicing mechanism applies for partitioning a
graph for use on a parallel computer --- we may seek separators that
require little information to be transferred from one processor to
another. A reasonable proxy for such ``compressibility'' is the
entropy of a set of random variables, a well-known submodular
function. The resulting optimization problem of finding a
minimum-entropy separator is a cooperative cut that is different from any known special cases.

\paragraph{Robust optimization.} Assume we are given a graph where the
weight of each edge $e \in \Es$ is noisy and distributed as $\mathcal{N}(\mu(e),
\sigma^2(e))$ for nonnegative mean weights $\mu(e)$. The noise on different edges is independent, and the
cost of a cut is the sum of edge weights of an unknown draw
from that distribution. In such a case, we might want to not only
minimize the expected cost, but also take the variance into
consideration. This is the aim in mean-risk minimization
(which is equivalent to a probability tail model or value-at-risk
model), where we aim to minimize
\begin{equation}
  f(C) = \sum_{e \in C}\mu(e) + \lambda \sqrt{\sum_{e \in C}\sigma^2(e)}.
\label{eq:mean_risk_minimization}
\end{equation}
This is a cooperative cut, and this special case admits an FPTAS \citep{nik10}.

\subsection{Special cases that lead to submodular functions $h$}
Curiously, 
some instances of cooperative cuts provably yield submodular node functions $h$ and are hence easier. In the first two examples below, $f$ is defined over edges in a complete graph and is \emph{symmetric}. Here, symmetry is meant in the sense of \citep{vondrak2011} and Proposition~\ref{prop:symm}, the function is indifferent to permutations of the arguments.

\paragraph{Higher-order potentials in computer vision.} A number of
higher-order potentials (pseudo-boolean functions) from computer vision, i.e., potential functions that introduce dependencies between more than two variables at a time, can be
reformulated as cooperative cuts. As an example, $P^n$ Potts functions
\citep{kohli2009p} 
and robust $P^n$ potentials \citep{kohli2009robust} bias image labelings to assign the same label to larger patches of pixels (of uniform appearance). The potential is low if all nodes in a given patch take the same label, and high if a large fraction deviates from the majority label.
These potential functions
 correspond to a
complete graph with a 
cooperative cut cost function
\begin{equation}
  \label{eq:pn}
  f(C) = g(|C|),
\end{equation}
for a concave function $g$. The larger the fraction of deviating nodes, the more edges are cut between labels, leading to a higher penalty. The function $g$ makes this robust by capping the maximum penalty.

\paragraph{Regularization and Total Variation.} A popular
regularization term in signal processing, and in particular for image
denoising, has been the Total Variation (TV) and its discretization
\citep{rof92}. The setup commonly includes a pixel variable (say $x_j$
or $x_{ij}$) corresponding to each pixel or node in the image graph $\Gs$,
and an objective that consists of a loss term and the regularization.
The discrete TV for variables $x_{ij}$ corresponding to pixels
$v_{ij}$ in an $M \times M$ grid with coordinates $i,j$ is given as
\begin{equation}
  \mathrm{TV}_1(x) = \sum_{i,j=1}^M \sqrt{ (x_{i+1,j}-x_{ij})^2 + (x_{i,j+1}-x_{ij})^2}.
\end{equation}
If $x$ is constrained to be a $\{0,1\}$-valued vector, then this is an
instance of cooperative cut --- the pixels valued at unity correspond
to the selected elements $X \subseteq \Vs$, and the edge submodular
function corresponds to $f(C) = \sum_{ij} \sqrt{ C \cap S_{ij} }$ for $C
\subseteq \Es$ where  $S_{ij} = \{(v_{i+1,j}v_{ij}),(v_{i,j+1},v_{ij})\} \subseteq \Es$ ranges over all relevant neighborhood pairs of edges.
Discrete versions
of other variants of total variation are also cooperative cuts.
Examples include the combinatorial total variation of \cite{cgt11}:
\begin{equation}
  \mathrm{TV}_2(x) = \sum_i \sqrt{ \sum_{(v_i, v_j) \in \Es} \nu_i^2(x_i - x_j)^2},
\end{equation}
and the submodular oscillations in \citep{cd09}, for instance,
\begin{align}
  \nonumber
  \mathrm{TV}_3(x)  & = \sum_{1\leq i,j \leq M} \max\{x_{i,j}, x_{i+1,j},   x_{i,j+1}, x_{i+1, j+1}\} \\
   & \qquad \qquad \qquad 
  - \min \{x_{i,j}, x_{i+1,j}, x_{i,j+1},   x_{i+1, j+1}\}  \\
  \label{eq:tvfour}
  & =  \sum_{1\leq i,j \leq M} \max_{\ell, r \in U_{ij} \times U_{ij} } |x_\ell - x_k|,
\end{align}
where for notational convenience we used $U_{ij} = \{(i,j), (i+1,j),
(i, j+1), (i+1,j+1)\}$. The latter term \eqref{eq:tvfour}, like $P^n$ potentials,
corresponds to a 
symmetric submodular function on a complete graph, and
both \eqref{eq:pn} and \eqref{eq:tvfour}
lead to submodular node functions $h(X)$.

\paragraph{Approximate submodular minimization.} Graph cuts have been
useful optimization tools but cannot represent any arbitrary set
function, not even all submodular functions \citep{zcj09}.  But, using
a decomposition theorem by \citet{cun82}, \emph{any} submodular
function can be phrased as a cooperative graph cut. As a result, any
fast algorithm that computes an approximate minimum cooperative cut
can be used for (faster) approximate minimization
of certain submodular functions \citep{jlb11}.

\subsection{General Cooperative Cut}
\label{sec:gener-coop-cut}

The above examples indicate that certain special cases reduce \mcoocut{} to a submodular minimization problem, or result in a simpler optimization problem than the general form of \mcoocut{} with an arbitrary non-negative submodular cost function $f$ and an arbitrary graph $\Gs$. We will discuss such examples in further detail in Section~\ref{sec:easier}.

Yet, there are many reasons for studying the optimization landscape of general \mcoocut. First, not all of the examples in Section~\ref{subsec:nonsubmod-ex} fall into one of the ``easier'' classes. Second, applications often benefit from \emph{learning} the submodular function rather than specifying it a priori. While learning a submodular function is hard in general \citep{ghim09,bh12}, learning can be practically viable for sub-classes of submodular functions \citep{linB12learning,fix13,sk12,tschiatschek14}. Applications such as those in computer vision \citep{jb11} would likely benefit from learning too, but the resulting cooperative cut problem would not necessarily fall into an easy sub-class. Moreover, applications often benefit from combining different objective terms. In computer vision, this may be a cooperative cut potential for encouraging object boundaries of homogeneous appearance, combined with terms that express the data likelihood for different object classes, terms that encourage the coherence of uniform patches of pixels, e.g. the potentials in \citep{kohli2009robust}, and possibly others.
All of these terms are cooperative cuts, but together they quickly exceed special sub-classes of the problem.

In fact, empirical results enabled by general algorithms may hint at the existence of further, easier special cases that help map the complexity landscape. The empirical results in Section~\ref{sec:expt}, for example, are based on the results in this paper and open up questions for further study.
Hence, this paper deliberately takes a general viewpoint to connect the many examples from a spectrum of areas to a common optimization problem.

\section{Complexity and lower bounds}\label{sec:complexity}
In this section, we 
address the hardness of the general
\mcoocut{} problem. Assuming that the cost function is given as an
oracle, we show a lower bound of $\Omega(\sqrt{n})$ on the 
approximation factor. 
In addition, we include a proof of NP-hardness.
NP-hardness holds even if the cost function is completely known
and polynomially computable and representable.

Our results complement known lower bounds for related
combinatorial problems having submodular cost
functions. Table~\ref{tab:lbgen} provides an overview of known results
from the literature. 
In addition, \citet{zctz09} show a lower bound for the special case of
{\sc Minimum Label Cut} via a reduction from {\sc Minimum Label
  Cover}. Their lower bound is $2^{(\log \bar{m})^{1-(\log\log
    \bar{m})^{-c}}}$ for $c<0.5$, where $\bar{m}$ is the input length
of the instance. Their proof is based on the PCP theorem. In contrast,
the proof of the lower bound in Theorem~\ref{thm:lower} is
information-theoretic.

\begin{table}
  \centering
  {\small
    \renewcommand{\arraystretch}{1.2} 
    \begin{tabular}{|p{0.31\textwidth}|c|p{0.31\textwidth}|}
      \hline
      Problem & $\;\;\;$Lower Bound$\;\;\;$ & Reference \\ \hline
      Set Cover & $\Omega(\ln|\mathcal{U}|)$
      & \cite{in09b}\\ 
      Minimum Spanning Tree & $\Omega(n)$ & \cite{gktp09}\\
      Shortest Path & $\Omega(n^{2/3})$ & \cite{gktp09} \\
      Perfect Matching & $\Omega(n)$ & \cite{gktp09}\\
      Minimum Cut & $\Omega(\sqrt{n})$ & Theorem \ref{thm:lower}\\ 
      \hline
    \end{tabular}
  }
  \caption[Hardness results for combinatorial problems with submodular
  costs]{Hardness results for combinatorial problems with submodular
    costs, where $n$ is the number of nodes, and $\mathcal{U}$ the
    universe to cover. These results assume oracle access to the cost function.}
  \label{tab:lbgen}
\end{table}

\begin{theorem}\label{thm:lower}
  No polynomial-time algorithm can solve \mcoocut{} with
  an approximation factor of $o(\sqrt{|\Vs|/\log |\Vs|})$. 
\end{theorem}

The proof relies on constructing two submodular cost functions $f$,
$h$ that are almost indistinguishable, except that they have quite
differently valued minima. In fact, with high probability they cannot
be distinguished with a polynomial number of function queries.
If the optima of $h$ and $f$ differ by a factor larger than $\alpha$,
then any solution for $f$ within a factor $\alpha$ of the optimum
would be enough evidence to discriminate between $f$ and $h$. As a result, a
polynomial-time algorithm that guarantees an approximation factor
$\alpha$ would lead to a contradiction. The proof technique follows along the lines of the proofs in \citep{ghkm,ghim09,sf08}.

One of the functions, $f$, depends on a hidden random set $R \subseteq
E$ that will be its optimal cut. We will use the following lemma that
assumes $f$ will depend on a random set $R$.

\begin{lemma}[\cite{sf08}, Lemma~2.1]\label{lem:prob}
  If for any fixed set $Q \subseteq \Es$, chosen without knowledge of $R$,
  the probability of $f(Q) \neq h(Q)$ over the random choice of $R$ is
  $m^{-\omega(1)}$,
  then any algorithm that makes a polynomial number of oracle
  queries has probability at most $m^{-\omega(1)}$ of distinguishing
  between $f$ and $h$.  \looseness-1
\end{lemma}
Consequently, the two functions $f$ and $h$ in Lemma~\ref{lem:prob} cannot be distinguished with high probability within a polynomial number of queries, i.e., within polynomial time.
Hence, it suffices to construct two functions for which Lemma~\ref{lem:prob} holds.

\begin{proof}[Theorem~\ref{thm:lower}]
  We will prove the bound in terms of the number $m = |\Es|$ of edges
  in the graph. The graph we construct has $n = m - \ell + 2$ nodes, and
  therefore the proof also shows the lower bound in terms of nodes.

  Construct a graph $\Gs = (\Vs,\Es)$ with
  $\ell$ parallel disjoint paths from $s$ to $t$, where each path has $k$
  edges. The random set $R 
  \subset \Es$ is always a cut consisting of $|R|=\ell$ edges, and contains
  one edge from each path uniformly at random. We define $\beta =
  8\ell / k < \ell$ (for $k > 8$),
  and, for any $Q \subseteq \Es$,
  \begin{align}
    h(Q) &= \min\{ |Q|, \ell\} \\ 
    f(Q) &= \min\{ |Q \setminus R| + \min\{ |Q \cap R|,\;
    \beta \}, \; \ell\}.
  \end{align}
\begin{figure}[t]
    \centering
    \includegraphics[width=0.45\textwidth]{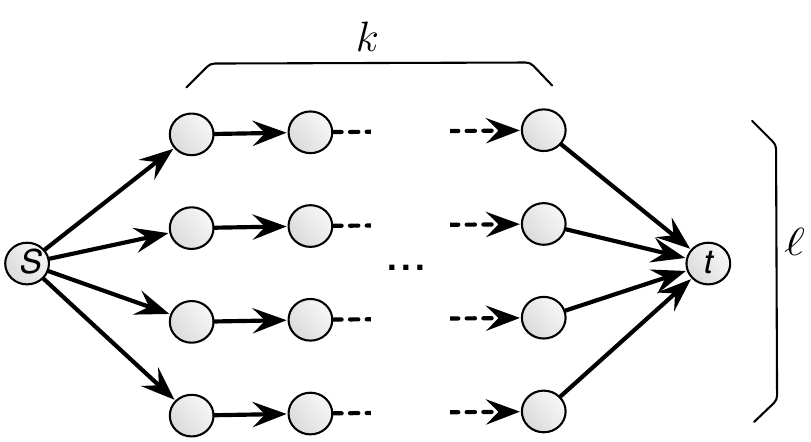} 
    \caption{Graph for the proof of Theorem \ref{thm:lower}.}
    \label{fig:lbgraph}
\end{figure}%
The functions differ only for the relatively few sets $Q$ with $|Q
\inter R| > \beta$ and $|Q \setminus R| < \ell-\beta$, with $\min_{A
  \in \Cs } h(A) = h(C) = \ell$, $\min_{A \in \Cs} f(A) = f(R) =
\beta$, where $\Cs$ is the set of cuts, and $C$ is any cut.  We must have $k\ell = m$, so define
$\epsilon$ such that $\epsilon^2 = 8/7 \log m$, and set $k =
8\sqrt{m}/\epsilon$ and $\ell = \epsilon\sqrt{m}/8$.

We compute the probability that $f$ and $h$ differ for a given query
set $Q$. Probabilities are over the random unknown $R$. Since $f \leq
h$, the probability of a difference is $\pr(f(Q) < h(Q))$.  If $|Q| \leq
\ell$, then $f(Q) < h(Q)$ only if $\beta < |Q \cap R|$, and the
probability $\pr(f(Q) < h(Q)) = \pr(|Q \cap R| > \beta)$ increases as $Q$
grows. If, on the other hand, $|Q| \geq \ell$, then 
since $h(Q) = \ell$ the probability
  \begin{equation*}
    \pr(f(Q) < h(Q)) 
    = \pr( |Q \setminus R| + \min\{ |Q \cap R|,\;
    \beta\} < \ell)
    = \pr(| Q \cap R | > \beta )
  \end{equation*}
  decreases as $Q$ grows. Hence, the probability of a difference is
  largest when $|Q|~=~\ell$.

  So let $|Q| = \ell$.  If $Q$ spreads over $b \leq k$ edges of a path
  $P$, then the probability that $Q$ includes the edge in $P \inter R$
  is $b/k$.  The expected overlap between $Q$ and $R$ is the sum of hits on all paths:
  $\mathbb{E}[~|Q \cap R|~] = |Q|/k = \ell/k$.  Since the edges in $R$
  are independent across different paths, we may bound the probability of
  a large intersection by a Chernoff bound (with $\delta = 7$ in [\citenum{mu05}]):
  \begin{align}
    \pr\big(\, f(Q) \neq h(Q)\big)\, &\leq\, \pr\big(\,|Q \cap R| \geq
    8\ell / k\;\big) \\
    &\leq\; 2^{- 7 \ell / k}    
    \; = \; 2^{-7\epsilon^2/8} = 2^{-\omega(\log m)} = m^{-\omega(1)}.
  \end{align}
  With this result, Lemma~\ref{lem:prob} applies. No polynomial-time
  algorithm can guarantee to be able to distinguish $f$ and $h$ with high
  probability. A polynomial algorithm with approximation factor better
  than the ratio of optima $h(R)/f(R)$ would discriminate the two functions and
  thus lead to a contradiction.
  As a result, the lower bound is determined by the ratio of optima of $h$ and
  $f$. The optimum of $f$ is $f(R)=\beta$, and $h$ has uniform cost
  $\ell$ for all minimal cuts. Hence, the ratio is $h(R)/f(R) = \ell /
  \beta = \sqrt{m}/\epsilon = o(\sqrt{m/\log m})$. \qed
\end{proof}

Building on the construction in the above proof with $\ell = n^{1/3}$ and a
different cut cost function, \citet{bh12} proved that if the data structure
used by an algorithm (even with an arbitrary number of queries) has
polynomial size, then this data structure cannot represent the
minimizers of their cooperative cut problem to an approximation factor
of $o(n^{1/3}/\log n)$.

In addition, we mention that a reduction from Graph Bisection serves
to prove that \mcoocut{} is NP-hard. We defer the proof to
Appendix~\ref{sec:reduct}, but point out that in the reduction, the
cost function is fully accessible and given as a polynomial-time
computable formula.
\begin{theorem}\label{thm:np}
  Minimum Cooperative $(s,t)$-Cut is NP-hard.
\end{theorem}

\section{Relaxation and the flow dual}\label{sec:relax}
As a first step towards approximation algorithms, we formulate a
relaxation of \mcoocut{} and analyze the flow-cut gap. 
The minimum cooperative cut problem can be relaxed to a continuous convex optimization problem using the convex Lov\'asz extension $\ft$ of $f$:
\begin{align}
  \label{eq:relaxshort}
  \min_{y \in \mathbb{R}^{|\Es|},\, x \in \mathbb{R}^{|\Vs|}}\;\;& \ft(y)\\
  \nonumber
  \text{s.t.} \quad -x(u) + x(v) + y(e) &\geq 0 \quad \text{for all }
  e=(u,v) \in \Es\\ 
  \nonumber
  x(s) - x(t) &\geq 1\\
  \nonumber
  y &\geq 0
\end{align}
The dual of this problem can be derived by writing the Lov\'asz
extension as a maximum $\ft(y) = \max_{z \in \mathcal{P}(f)} z^\top y$ of linear functions. The maximum is taken over the submodular polyhedron
\begin{equation}
  \mathcal{P}(f) = \{y \mid \sum_{e \in A} y(e) \leq f(A)\; \forall A \subseteq \Es\}.
\end{equation}
The resulting dual problem 
is a flow problem with
non-local capacity constraints: 
\begin{align}
  \label{eq:relflow}
  \max_{\nu \in \R,\varphi \in \R^{|\Es|}}\quad&\; \nu&\\
  \label{eq:constr1}
  \text{s.t. }\quad\;\; \varphi(A) \triangleq \sum_{e \in A} \varphi(e) \;&\leq\; f(A) 
  \;\;\;\text{ for all } A \subseteq \Es\\ 
    \nonumber
  \sum_{e \in \delta^+ u } \varphi(e) - \sum_{e' \in \delta^-u}
  \varphi(e') \;&=\; d(u)\nu 
  \quad \text{ for all } u \in \Vs\\
  \nonumber
  \varphi &\geq 0,
\end{align}
where 
$d(u) = 1$ if $u=s$, $d(u)=-1$ if $u=t$, and $d(u) = 0$ otherwise. Constraint~\eqref{eq:constr1} demands that $\varphi$ must, in addition to satisfying the common flow conservation, reside within the submodular polyhedron $\mathcal{P}(f)$. This more restrictive constraint replaces the edge-wise capacity constraints that occur when $f$ is a sum of weights.

As an alternative to \eqref{eq:relaxshort}, the constraints
can be stated in terms of paths: a set of edges is a cut if it
intersects all $(s,t)$-paths in the graph.
\begin{align}
\label{eq:cover}
  \min\; &\;\ft(y)\\
  \nonumber
  \text{s.t. }& \nlsum_{e \in P}y(e) \geq 1 \;\;\; \text{for all }
  (s,t)\text{-paths } P\subseteq \Es\\
  \nonumber
  & y \in [0,1]^\Es.
\end{align}
We will use this form in Section~\ref{sec:randgreed}, and the relaxation (\ref{eq:relaxshort}) in Section~\ref{sec:rounding}.

\subsection{Flow-cut gap}\label{subsec:flowcutgap}
\label{sec:fcgap}

The relaxation \eqref{eq:relaxshort} of the discrete problem \eqref{eq:cc} is not tight. 
This becomes evident when analyzing the ratio $f(C^*)/\ft(y^*)$ between the optimal value of the discrete problem and the relaxation \eqref{eq:relaxshort} (i.e., the \emph{integrality gap}). This ratio is, by strong duality between Problems \eqref{eq:relaxshort} and \eqref{eq:relflow}, also the \emph{flow-cut gap} $f(C^*) / \nu^*$ of the optimal cut and maximal flow values.
\begin{lemma}\label{lem:gapbds}
  Let $\mathcal{P}$ be the set of all $(s,t)$-paths in the graph.
  The flow-cut gap $f(C^*)/\nu^*$
  can be upper and lower bounded as follows:
  \begin{equation*}
    \frac{f(C^*)}{\sum_{P \in \mathcal{P}}\min_{P' \subseteq
        P}\frac{f(P')}{|P'|}} \leq \frac{f(C^*)}{\nu^*} \leq
    \frac{f(C^*)}{\max_{P \in \mathcal{P}}\min_{P' \subseteq P} \frac{f(P')}{|P'|}}.
  \end{equation*}
\end{lemma}
\begin{proof}
  The Lemma straightforwardly follows from bounding the optimal flow $\nu^*$. The flow through a single path $P \in \mathcal{P}$, if all other edges $e \notin \mathcal{P}$ are empty, is restricted by the minimum average capacity for any subset of edges within the path, i.e., $\min_{P' \subseteq P} \frac{f(P')}{|P'|}$. Moreover, we obtain a family of feasible solutions as those that send nonzero flow only along one path and remain within that path's capacity. Hence, the maximum flow must be at least as big as the flow for any of those single-path solutions. This observation yields the upper bound on the ratio. 

  A similar argumentation shows the lower bound: the total joint capacity constraint is upper bounded by $\hat{f}(A) = \sum_{P \in \mathcal{P}}f(A \inter P) \geq f(A)$. Hence, $\sum_{P \in \mathcal{P}}\min_{P' \subseteq P}\frac{f(P')}{|P'|}$ is the value of the maximum flow with capacity $\hat{f}$ if each edge is only contained in one path, and is an upper bound on the flow otherwise. \qed
\end{proof}

\begin{corollary}\label{lem:fcgap}
  The flow-cut gap for \mcoocut{} can be as large as $n-1$.
\end{corollary}
\begin{proof}
  Corollary~\ref{lem:fcgap} can be shown via an example where the upper and lower bound of Lemma~\ref{lem:gapbds} coincide.
  The worst-case example for the flow-cut gap is a simple
graph that consists of one single path from $s$ to $t$ with $n-1$ edges. 
For this graph one of the capacity constraints is that
\begin{equation}
  \label{eq:flow_constraint}
  \varphi(\Es) = \sum_{e \in \Es}\varphi(e) \leq f(\Es).
\end{equation}
Constraint~\eqref{eq:flow_constraint} is the only relevant capacity constraint if 
the capacity (and cut cost)
function is $f(A) = \max_{e \in A}w(e)$ with weights $w(e) = \gamma$
for all $e \in \Es$ and some constant $\gamma>0$ and, consequently, $f(\Es) = \gamma$. By Constraint~\eqref{eq:flow_constraint}, the maximum flow is $\nu^* =
\frac{\gamma}{n-1}$. The optimum cooperative cut $C^*$, by contrast,
consists of any single edge and has cost $f(C^*) = \gamma$. \qed
\end{proof}
  Single path graphs as used in the previous proof can provide
  worst-case examples for rounding methods too: if $f$ is such that $f(e)
  \geq f(\Es)/|\Es|$ for all edges $e$ in the path, then the solution to
  the relaxed cut problem is maximally uninformative: all entries of
  the vector $y$ are $y(e) = \frac{f(\Es)}{n-1}$.

\section{Approximation algorithms}\label{sec:algo}

We next address approximation algorithms whereby we consider two
complementary approaches.  The first approach is to substitute the
submodular cost function $f$ by a simpler function $\hat{f}$.
Appropriate candidate functions $\hat{f}$ that admit an exact cut
optimization are the approximation by \citet{ghim09}
(Section~\ref{subsec:generic}), semi-gradient based approximations
(Section~\ref{subsec:semigrad}), or approximations by making $f$
separable across local neighborhoods (Section~\ref{subsec:pmf}).

The second approach is to solve the relaxations from Section~\ref{sec:relax} and round the resulting optimal fractional solution (Section~\ref{sec:rounding}). Conceptually very close to the relaxation approach, we offer another algorithm that solves the mathematical program~(\ref{eq:cover}) via a randomized greedy algorithm (Section~\ref{sec:randgreed}).

The relaxations approaches are affected by the flow-cut gap, or, equivalently, the length of the longest path in the graph. The approximations that use a surrogate cost function are complementary and not affected by the ``length'', but by a notion of the ``width'' of the graph.

\begin{table}
  \centering
  \small
  \begin{tabular}{|p{0.33\textwidth}c|p{0.28\textwidth}c|}
    \hline
    \multicolumn{2}{|c|}{approximating $f$} &     \multicolumn{2}{c|}{relaxation} \\ \hline
    generic (\S \ref{subsec:generic}) & $O(\sqrt{m}\log m)$ &
    randomized (\S \ref{sec:randgreed}) & $|P_{\max}|$
    \\
    semigradient (\S \ref{subsec:semigrad}) & $\frac{|C^*|}{(|C^*|-1)(1-\kappa_f) + 1}$ &
    rounding I (\S \ref{sec:rounding}) & $|P_{\max}|$
    \\
    polymatroidal flow (\S \ref{subsec:pmf}) & $\min\{\Delta_s, \Delta_t\}$ &
    rounding II (\S \ref{sec:rounding}) & $|\Vs|-1$
    \\
    \hline
  \end{tabular}
  \caption{Overview of the algorithms and their approximation factors.}
  \label{tab:algos_all}
\end{table}

\subsection{Approximating the cost function}\label{sec:surrog}
We begin with algorithms that use a suitable approximation $\fa$ to
$f$, for which the problem
\begin{equation}
  \text{minimize } \fa(C) \quad \text{s.t. $C \subseteq \Es$ is a cut}
\end{equation}
is solvable exactly in polynomial time. The following lemma will be
the basis for the approximation bounds.
\begin{lemma}\label{lem:fapprox}
  Let $\Sh = \argmin_{S \in \mathcal{S}}\fa(S)$.
  If for all $S \subseteq \Es$, it holds that $f(S) \leq \fa(S)$, and
  if for the optimal solution $S^*$ to Problem~(\ref{eq:cc}), it
  holds that $\fa(S^*) \leq \alpha f(S^*)$, then $\Sh$ is an $\alpha$-approximate solution to 
  Problem~(\ref{eq:cc}):
  \begin{equation*}
    f(\Sh)\; \leq\; \alpha f(S^*).
  \end{equation*}
\end{lemma}
\begin{proof}
  Since $\fa(\Sh) \leq \fa(S^*)$, it follows that
  $f(\Sh) \leq \fa(\Sh) \leq \fa(S^*) \leq \alpha f(S^*)$. \qed
\end{proof}

\subsubsection{A generic approximation}\label{subsec:generic}
\citet{ghim09} define a generic approximation of a monotone submodular function\footnote{We will also call it the \emph{ellipsoidal approximation} since it is based on approximating a symmetrized version of the submodular polyhedron by an ellipsoid.}
that has the functional form $\fae(A) = \sqrt{\sum_{e \in A}w_f(e)}$. The
  weights $w_f(e)$ depend on $f$. 
  When using $\fae$, we
compute a minimum cut for the cost $\fae^2$, which is a modular sum of weights and hence results in 
a standard \mcut{} problem.
In practice, the bottleneck lies
in computing the weights $w_f$.
\citet{ghim09} show how to compute weights such that $f(A) \leq \fa(A) \leq \alpha f(A)$ with
$\alpha = O(\sqrt{m})$ for a matroid rank function, and $\alpha =
O(\sqrt{m}\log m)$ otherwise. We add that for an integer polymatroid rank function bounded by $M =
\max_{e \in \Es} f(e)$, the logarithmic factor can be replaced by a
constant to yield $\alpha = O(\sqrt{mM})$ (if one approximates the matroid
expansion\footnote{The expansion is described in Section~10.3 in
  \citep{na97}. In short, we replace each element $e$ by a set
  $\hat{e}$ of $f(e)$ parallel elements. Thereby we extend $f$ to a
  submodular function $\hat{f}$ on subsets of $\bigcup_i \hat{e}_i$. The desired
  rank function is now the convolution $r(\cdot) = \hat{f}(\cdot) \ast
  |\cdot|$ and it satisfies $f(S) = r(\bigcup_{e \in S}\hat{e})$.} of
the polymatroid instead of $f$ directly). 
Together with Lemma~\ref{lem:fapprox}, this yields the following
approximation bounds.
\begin{lemma}\label{lem:ubellipsoid}
  Let $\Ch = \argmin_{C \in \Cs}\fae(C)$ be the minimum cut for cost
  $\fae$, and $C^* = \argmin_{C \in \Cs}f(C)$. Then $f(\widehat{C}) =
  O(\sqrt{m}\log m) f(C^*)$. If $f$ is integer-valued and we
  approximate its matroid expansion, then $f(\widehat{C}) =
  O(\sqrt{mM}) f(C^*)$, where $M \leq \max_ef(e)$.
\end{lemma}
The lower bound in Theorem~\ref{thm:lower} suggests that for sparse graphs, the bound in
Lemma~\ref{lem:ubellipsoid} is tight up to logarithmic factors.

\subsubsection{Approximations via semigradients}
\label{subsec:semigrad}

For any monotone submodular function $f$ and any set $A$, there is a simple way to compute a modular upper bound $\fs$ to $f$ that agrees with $f$ at $A$. In other words, $\fs$ is a discrete \emph{supergradient} of $f$ at $A$. We define $\fs$ as \citep{jb11,iyer13} 
\begin{align}
  \fs(B; A) = f(A) + \sum_{e \in B \setminus A} f(e \mid A) - \sum_{e \in A \setminus B} f(e \mid \Es \setminus e).
\end{align}
\begin{lemma}\label{lem:curv}
  Let $\widehat{C} \in \argmin_{C \in \Cs}\fs(C; \emptyset)$. Then 
  \begin{equation*}
    f(\widehat{C}) \leq \frac{|C^*|}{(|C^*|-1)(1-\kappa_f) + 1} f(C^*),
  \end{equation*}
  where $\kappa_f = \max_{e}\, \big(1- \frac{f(e \mid \Es \setminus e)}{f(e)}\big)$ is the curvature of $f$.
\end{lemma}
Lemma~\ref{lem:curv} was shown in \citep{iyer13}. 
As $m$ (and
correspondingly $|C^*|$) gets large, the bound eventually no longer
depends on $m$ and instead only on the curvature of $f$.  
In practice, results are best when 
the supergradient is used in an iterative algorithm: starting with $C_0 =
\emptyset$, one computes $C_t \in \argmin_{C \in \Cs}\fs(C; C_{t-1})$ until
the solution no longer changes between iterations. The minimum cut for
the cost function $\fs(C; A)$ can be computed as a minimum cut with
edge weights
\begin{equation}
  w(e) =
  \begin{cases}
    f(e \mid \Es \setminus e) & \text{ if } e \in A\\
    f(e \mid A) & \text{ if } e \notin A.
  \end{cases}
\end{equation}
Consequently, the semigradient approximation yields a very easy and practical algorithm
that iteratively uses standard minimum cut as a subroutine.
This algorithm was used e.g.\ in \citep{jb11}, and the visual results in \citep{kohliOJ13} show that it typically yields very good solutions in practice on 
certain problem instances where the optimum solution can be computed
exactly.

\subsubsection{Approximations by introducing separation}\label{subsec:pmf}

The approximations in Section~\ref{subsec:generic} and
\ref{subsec:semigrad} are indifferent to the structure of the
graph, while following approximation is not.  One may say that
Problem~(\ref{prob:cc}) is hard because $f$ introduces non-local
dependencies between edges that might be anywhere in the graph.
Indeed, the problem is easier if dependencies between edges are
restricted to local neighborhoods defined by the graph, for example,
edges that might be incident to the same vertex.

\begin{figure}
  \centering
  \begin{minipage}{0.45\linewidth}
    \begin{center}
      \includegraphics[width=0.95\textwidth]{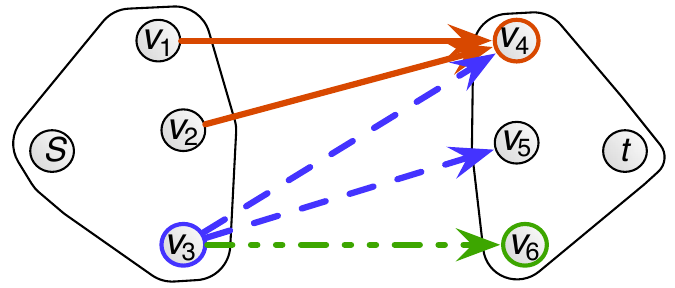}
    \end{center}
  \end{minipage}
  \begin{minipage}{0.45\linewidth}
    \begin{align*}
      \fap(C) = &f(\{(v_1,v_4),(v_2,v_4)\})\\
      &+ f(\{(v_3,v_4),(v_3,v_5)\})\\
      &+ f(\{(v_3,v_6)\})
    \end{align*}
  \end{minipage}
  \caption[Approximation of the cut cost via a partition.]{Approximation of a cut cost. Red edges are in $C^{\Pi}_{v_4}$ (head), blue dashed edges in $C^{\Pi}_{v_3}$ (tail), and the green dash-dotted edge in $C^{\Pi}_{v_6}$ (head).
}
  \label{fig:lawler}
\end{figure}

Hence, we define an approximation $\fap$ that is globally
separable but locally exact. To measure the cost of
an edge set $C \subseteq \Es$, we partition $C$ into groups $\Pi(C)=\{C^\Pi_v\}_{v \in
   V}$, where the edges in set $C^\Pi_v$ must be incident
 to node $v$ ($C_v^\Pi$ may be empty). That is, we assign each edge either to its
 head or to its tail node in any partition, as illustrated in Figure~\ref{fig:lawler}. Let $\mathcal{P}_C$
 be the family of all such partitions (which vary over the head or tail
 assignment of each edge). We define an approximation 
\begin{equation}
  \label{eq:pmf}
  \fap(C) = \min_{\Pi(C) \in \mathcal{P}_C}\nlsum_{v \in V}f(C^\Pi_v)
\end{equation}
that (once the partition is fixed) decomposes across different node
incidence edge sets, but is accurate within a group $C^\Pi_v$.  Thanks
to the subadditivity of $f$, the function $\fap$ is an upper bound on
$f$. It is a convolution of submodular functions and always is the tightest approximation that is a direct sum over
any partition in $\mathcal{P}_C$.  Perhaps surprisingly, even though the
approximation~\eqref{eq:pmf} looks difficult to compute and is in
general not even a submodular function (an example is in
Appendix~\ref{app:convolution}), it is possible to solve a
minimum cut with cost $\fap$ exactly. To do so, we exploit its duality
to a generalized maximum flow problem, namely polymatroidal network flows.

\paragraph{Polymatroidal network flows.}
Polymatroidal network flows \citep{lm82,h82} generalize the capacity constraint of
traditional flow problems. They retain the constraint of flow conservation (%
a function $\varphi: \Es \to \mathbb{R}_+$ is
a flow if the inflow at each node $v \in \Vs \setminus\{s,t\}$ equals the
outflow).
The edge-wise capacity constraint 
$\varphi(e) \leq \mathrm{cap}(e)$ for all $e \in \Es$, given a capacity function
$\mathrm{cap}: \Es \to \mathbb{R}_+$
is replaced by
local submodular capacities over
\emph{sets} of edges incident at each node $v$:
$\mathrm{cap}_v^{\text{in}}$ for incoming edges, and 
$\mathrm{cap}_v^{\text{out}}$ for outgoing edges. 
The capacity constraints at each $v \in \Vs$ are
  \begin{align*}
    \varphi(A) &\leq \mathrm{cap}_v^{\text{in}}(A) \quad\;\, \text{ for all } A \subseteq \delta^-(v) \; \text{(incoming edges), and}\\
    \varphi(A) &\leq \mathrm{cap}_v^{\text{out}}(A) \quad \text{ for all } A \subseteq \delta^+(v) \; \text{(outgoing edges)}.
  \end{align*}
Each edge $(u,v)$ belongs to two incidence sets, $\delta^+u$ and $\delta^-v$.
A maximum flow with such constraints can be found in time $O(m^4\tau)$ by the layered augmenting path algorithm by
\citet{ttt86}, where $\tau$ is the time to minimize a submodular function on any set $\delta^+v$, $\delta^-v$ for any $v$. Hence, the 
incidence sets are in general much smaller than $\Es$.

A special polymatroidal maximum flow turns out to be dual to the cut
problem we are interested in. To see this, we will use the restriction
$\restriction{f}{A}$ of the function $f$ to a subset $A$. 
For ease of reading we drop the explicit restriction notation later. We assume throughout that the desired cut is minimal\footnote{A cut $C \subseteq \Es$ is \emph{minimal} if no proper subset $B \subset C$ is a cut.}, since additional
edges can only increase its cost.
\begin{lemma}\label{lem:dualpoly}
  Minimum $(s,t)$-cut with cost function $\fap$ is dual to a
  polymatroidal network flow with capacities
  $\mathrm{cap}_v^{\text{in}}=\restriction{f}{\delta^-v}$ and
  $\mathrm{cap}_v^{\text{out}}=\restriction{f}{\delta^+v}$ at each
  node $v \in \Vs$.
\end{lemma}
The proof is provided in Appendix~\ref{app:dualPMF}. It uses,
with some additional considerations, the
dual problem to a polymatroidal maxflow, which can be stated as
follows. Let $\mathrm{cap}^{\text{in}}: 2^\Es \to \mathbb{R}_+$ be the
joint incoming capacity function, i.e., $\mathrm{cap}^{\text{in}}(C) = \sum_{v \in
  V}\mathrm{cap}_v^{\text{in}}(C \inter \delta^-v)$, and let equivalently
$\mathrm{cap}^{\text{out}}$ be the corresponding joint outgoing capacity. The dual of
the polymatroidal maximum flow is a minimum cut problem whose cost is a
convolution of edge capacities \citep{lo83}:
\begin{equation}\label{eq:capacity_conv}
  \mathrm{cap}(C) =
  (\mathrm{cap}^{\text{in}} \ast \mathrm{cap}^{\text{out}})(C)
  \;\triangleq\; \min_{A \subseteq C} \Bigl[ \mathrm{cap}^{\text{in}}(A) +
  \mathrm{cap}^{\text{out}}(C \setminus A)\Bigr].
\end{equation}  
This convolution is in general not a submodular function.
Lemma~\ref{lem:dualpoly} implies that we can solve the approximate \mcoocut{} via its dual flow problem. The primal cut solution will be given by a set of full edges, i.e., edges whose joint flow equals their joint capacity.

We can now state the resulting approximation bound for \mcoocut. Let $C^*$ be the optimal cut for cost $f$. We define
$\Delta_s$ to be the tail nodes of the edges in $C^*$: 
$\Delta_s = \{v \mid
\exists (v,u)\in C^*\}$, and
similarly, $\Delta_t= \{v \mid \exists (u,v)\in C^*\}$.
The sets $\Delta_s, \Delta_t$ provide a measure of the ``width'' of the graph.
\begin{theorem}\label{thm:approxPMF}
  Let $\Ch$ be the minimum cut for cost $\fap$, and $C^*$ the
  optimal cut for cost $f$. Then
  \begin{equation*}
    f(\Ch)\; \leq\;
    \min\{|\Delta_s|,|\Delta_t|\}\, f(C^*)\; \leq\; \frac{|\Vs|}{2} f(C^*).
  \end{equation*}
\end{theorem}
\begin{proof}
  To apply Lemma \ref{lem:fapprox}, we need to show that $f(C) \leq
  \fap(C)$ for all $C \subseteq \Es$, and find an $\alpha$ such that
  $\fap(C^*) \leq \alpha f(C^*)$. The first
  condition follows from the subadditivity of $f$. 

  To bound
  $\alpha$, we use Lemma~\ref{lem:dualpoly} and Equation~\ref{eq:capacity_conv}:
  \begin{align}
    \fap(C^*) &= (\mathrm{cap}^{\text{in}} \ast
    \mathrm{cap}^{\text{out}})(C^*)\\
    \label{eq:bdIa}
    &\leq \min \{\mathrm{cap}^{\text{in}}(C^*),\;
    \mathrm{cap}^{\text{out}}(C^*)\} \\ 
    &\leq \min\Big\{\nlsum_{v \in \Delta_s} f(C^* \inter \delta^+ v),\;\;
    \nlsum_{v \in \Delta_t} f(C^* \inter \delta^- v)\Big\}\\
    &\leq \min\Big\{ |\Delta_s| \max_{v \in \Delta_s} f(C^* \inter
    \delta^+ v),\;\; |\Delta_t|
    \max_{v \in \Delta_t} f(C^* \inter \delta^- v)\Big\}\\
    \label{eq:bdIb}
    &\leq \min\big\{\, |\Delta_s|,\; |\Delta_t|\,\big\}\; f(C^*).
  \end{align}
  Thus, Lemma \ref{lem:fapprox} implies an approximation bound
  $\alpha \leq \min\big\{\, |\Delta_s|,\; |\Delta_t|\,\big\} \leq
  |\Vs|/2$. \qed 
\end{proof}

  \citet{iyercurv13} show that the bound in Theorem~\ref{thm:approxPMF} can be tightened to $\frac{|\Vs|}{2 + (|\Vs|-2)(1-\kappa_f)}$ by taking into account the curvature $\kappa_f$ of $f$.

\subsection{Relaxations}
\label{sec:relaxations}

An alternative approach to approximating the edge weight function $f$ is to relax the cut constraints via the formulations \eqref{eq:cover} and \eqref{eq:relaxshort}. We analyze two algorithms: the first, a randomized algorithm, maintains a discrete solution, while the second is a simple rounding method.  Both cases remove the constraint that the cut must be minimal: any set $B$ is feasible that has a \emph{subset} $C \subseteq B$ that is a cut. Relaxing the minimality constraint makes the feasible set \emph{up-monotone} (equivalently up-closed). This is not major problem, however, since any superset of a cut can easily be pruned to a minimal cut while only, if anything, improving the solution due to the monotonicity of $f$.

\subsubsection{Randomized greedy covering}
\label{sec:randgreed}

The constraints in the path-based relaxation (\ref{eq:cover}) suggest that 
a minimum $(s,t)$-cut problem is also a min-cost cover problem:
a cut must
intersect or ``cover'' each $(s,t)$-path in the
graph. The covering formulation of the constraints in \eqref{eq:cover} clearly show the relaxation of the minimality constraint.
Algorithm~\ref{alg:greedy} solves a discrete variant of
the formulation (\ref{eq:cover}) and maintains a discrete $y \in \{0,1\}$, i.e., $y$ is eventually the \emph{indicator vector} of a cut.

Since a graph can have exponentially many $(s,t)$-paths, there can be exponentially many constraints. But all that is needed in the algorithm is to find a violated constraint, and this
is possible by
computing the shortest path $P_{\min}$, using $y$ as the (additive) edge lengths. If
$P_{\min}$'s weight is at least one, then $y$ is feasible. Otherwise, $P_{\min}$ defines a violated constraint in formulation \eqref{eq:cover}.

\begin{algorithm}[t]
  \begin{algorithmic}
    \STATE{\bf Input:} graph $\Gs=(\Vs,\Es)$, terminal nodes $s,t \in \Vs$, cost function $f: 2^\Es \to \mathbb{R}_+$
    \STATE $C = \emptyset$, $y = 0$
    \WHILE{$\sum_{e \in P_{\min}}y(e)<1$ for the current shortest path $P_{\min}$}
    \STATE choose $\beta$ within the interval $\beta \in (0, \min_{e \in P_{\min}}f(e|C)]$
    \FOR{$e$ in $P_{\min}$}
    \STATE with probability $\beta/f(e|C)$, set $C = C \union \{e\}$, $y(e)=1$.
    \ENDFOR
    \ENDWHILE
    \STATE prune $C$ to $C'$ and return $C'$
  \end{algorithmic}
  \caption{Greedy randomized path cover}\label{alg:greedy}
\end{algorithm}

Owing to the form of the constraints, we can adapt a randomized greedy
cover algorithm 
\citep{ky09} to Problem (\ref{eq:cover}) and obtain
Algorithm~\ref{alg:greedy}.  In each step, we compute the shortest
path with weights $y$ to find a possibly uncovered path. Ties are
resolved arbitrarily.  To cover the path, we randomly pick edges from
$P_{\min}$. The probability of picking edge $e$ is inversely
proportional to the marginal cost $f(e|C)$ of adding $e$ to the
current selection of edges\footnote{If $ \min_{e \in
    P_{\min}}f(e|C)=0$, then we greedily pick all such edges with
  zero marginal cost, because they do not increase the cost. Otherwise
  we sample as indicated in the algorithm.}.  We must also specify an
appropriate $\beta$. With the maximal allowed $\beta = \min_{e \in P_{\min}}f(e|C)$,
the cheapest edges are selected deterministically, and others
randomly.  
In that case, $C$ grows by at least one edge 
in each iteration, and the
algorithm terminates after at most $m$ iterations.

If the algorithm returns a set $C$ that is feasible but not a
\emph{minimal} cut, it is easy to prune it to a minimal cut
$C' \subseteq C$ without any additional approximation error, since
monotonicity of $f$ implies that $f(C') \leq f(C)$.  Such pruning can
for example be done via breadth-first search. Let $\Vs_s$ be the set
of nodes reachable from $s$ after the edges in $C$ have been
removed. Then we set $C' = \delta^+(\Vs_s)$. The set $C'$ must be a
subset of $C$,  since if there was an edge
$(u,v) \in C'\setminus C$, then $v$ would also be in $\Vs_s$, and then
$(u,v)$ cannot be in $C'$, a contradiction.

The approximation bound for Algorithm~\ref{alg:greedy} is the length of the
longest path, like that of the rounding
methods in Section~\ref{sec:rounding}. This is not a coincidence, since both algorithms essentially use the same relaxation.
\begin{lemma}\label{lem:greedy}
  In expectation (over the probability of sampling edges),
  Algorithm~\ref{alg:greedy} returns a solution $\Ch'$ with
  $\mathbb{E}[f(\Ch')] \leq |P_{\max}|f(C^*)$, where $P_{\max}$ is the
  longest simple $(s,t)$-path in $\Gs$.  
\end{lemma}
\begin{proof} Let $\Ch$ be the cut before pruning.
  Since $f$ is nondecreasing, it holds that
  $f(\Ch') \leq f(\Ch)$.
  By Theorem 7 in \citep{ky09}, a greedy randomized procedure like
  Algorithm~\ref{alg:greedy} yields in expectation an
  $\alpha$-approximation for a cover, where 
  $\alpha$ is the maximum number of variables in any constraint. Here, $\alpha$ is the maximum number of edges in
  any simple path, i.e., the length of the
  longest path. This implies that $\mathbb{E}[f(\Ch')] \leq
  \mathbb{E}[f(\Ch)] \leq |P_{\max}|f(C^*)$. \qed
\end{proof}

Indeed, randomization is important. Consider a 
deterministic algorithm that always
picks the edge with minimum marginal cost in the next path to cover. The solution $\Ch_d$ returned by this algorithm can be much worse. As an example, consider a graph consisting of a clique $\mathcal{V}$ of $n$ nodes, with nodes $s$ and $t$. Let $S \subseteq \Vs$ be a set of size $n/2$. Node $s$ is connected to all nodes in $S$, and 
node $t$ is connected to the clique only by a distinct node $v' \in \Vs \setminus S$ via edge $(v',t)$. Let the cost function be a sum of edge weights, $f(C) = \sum_{e \in C}w(e)$. Edge $(v',t)$ has weight $\gamma > 0$, all edges in $\delta^+(S)$ have weight $\gamma(1 - \epsilon)$ for a small $\epsilon > 0$, and all remaining edges have weight $\gamma(1 - \epsilon/2)$. The deterministic algorithm will return $\Ch_d = \delta^+(S)$ as the solution, with cost $\frac{n^2\gamma}{4}(1-\epsilon)$, which is by a factor of $|\Ch_d|(1-\epsilon) = \frac{n^2}{4}(1-\epsilon)$ worse than the optimal cut, $f(\{(v',t)\}) = \gamma$. Hence, for the deterministic variant of Algorithm~\ref{alg:greedy}, we can only show the following approximation bound:
\begin{lemma}\label{lem:ghbound}
  For the solution $\Ch_d$ returned by the greedy deterministic
  heuristic, it holds that $f(\Ch_d) \leq |\Ch_d| f(C^*)$. This approximation factor cannot be improved in general.
\end{lemma}
\begin{proof}
  To each edge $e \in \Ch_d$ assign the path $P(e)$ which it was chosen
  to cover. By the nature of the algorithm, it must hold that $f(e)
  \leq f(C^* \inter P(e))$, because otherwise an edge in $C^* \inter
  P(e)$ would have been chosen. Since $C^*$ is a cut, the set $C^* \inter
  P(e)$ must be non-empty. These observations imply that
  \begin{align*}
    f(\Ch_d)\, \leq\, \sum_{e \in \Ch}f(e)\, \leq\, \sum_{e \in
      \Ch_d}f(C^* \inter P(e))\, \leq\, |\Ch_d|\max_{e \in \Ch_d}f(C^*
    \inter P(e))\, \leq\, |\Ch_d|f(C^*).
  \end{align*}
  Tightness follows from the worst-case example described above. \qed
\end{proof}

\subsubsection{Rounding}\label{sec:rounding}
Our last approach is to solve the convex program~\eqref{eq:relaxshort} and round the continuous to a discrete solution.
We describe two types of rounding, each of which achieves a worst-case
approximation factor of $n-1$. This factor equals the general flow-cut
gap in Lemma~\ref{lem:fcgap}. 
Let $x^*, y^*$ be the optimal solution to the
relaxation \eqref{eq:relaxshort} (equivalently, to \eqref{eq:cover}). We assume w.l.o.g.\
that $x^* \in [0,1]^n$, $y^* \in [0,1]^m$.

\paragraph{Rounding by thresholding edge lengths.}
The first technique uses the edge weights $y^*$. We pick a threshold
$\theta$ and include all edges $e$ whose entry $y^*(e)$ is larger than
$\theta$. Algorithm~\ref{alg:round} shows how to select $\theta$,
namely the largest edge length
that when treated as a threshold yields a cut. 

\begin{algorithm}[t]
  \caption{Rounding procedure given $y^*$}\label{alg:round}
  \begin{algorithmic}
    \STATE order $\Es$ such that $y^*(e_1) \geq y^*(e_2) \geq \ldots
    \geq y^*(e_m)$ 
    \FOR{$i=1, \ldots, m$}
    \STATE let $C_i = \{e_j\;|\; y^*(e_j) \geq y^*(e_i)\}$
    \IF{$C_i$ is a cut}
    \STATE prune $C_i$ to $\widehat{C}$ and return $\widehat{C}$
    \ENDIF
    \ENDFOR
  \end{algorithmic}
\end{algorithm}

\begin{lemma}\label{lem:ubrelax}
  Let $\Ch$ be the rounded solution returned by
  Algorithm~\ref{alg:round}, $\theta$ the threshold at the
  last iteration $i$, and $C^*$ the optimal
  cut. Then
  \begin{equation*}
    f(\Ch)\; \leq\; \frac{1}{\theta} f(C^*)\; \leq\; |P_{\max}|f(C^*)\; \leq\; (n-1) f(C^*),
  \end{equation*}
  where $P_{\max}$ is the longest simple path in the graph.
\end{lemma}
\begin{proof}
The proof is analogous to that for covering problems \citep{in09b}.
In the worst case,
$y^*$ is uniformly distributed along the longest path, i.e., $y^*(e) =
|P_{\max}|^{-1}$ for all $e \in P_{\max}$ as $y^*$ must sum to at
least one along each path. Then
$\theta$ must be at least $|P_{\max}|^{-1}$ to include at least one of the edges
in $P_{\max}$. 
Since $\ft$ is nondecreasing like $f$ and also positively homogeneous,
it holds that 
\begin{align*}
  f(\widehat{C}) \leq f(C_i) = \ft(\chi_{C_i})\; \leq \ft(\theta^{-1}y^*)
  = \theta^{-1}\ft(y^*) \leq \theta^{-1}\ft(\chi_{C^*}) = \theta^{-1}f(C^*).
\end{align*}
The first inequality follows from monotonicity of $f$ and the fact
that $\widehat{C} \subseteq C_i$. Similarly, the relation between
$\ft(\chi_{C_i})$ and $\ft(\theta^{-1}y^*)$ holds because $\ft$ is
nondecreasing: by construction, $y^*(e) \geq \theta\chi_{C_i}(e)$ for
all $e \in \Es$, and hence $\chi_{C_i}(e) \leq
\theta^{-1}y^*(e)$. Finally, we use the optimality of $y^*$ to relate
the cost to $f(C^*)$; the vector $\chi_{C^*}$ is also feasible, but $y^*$ optimal. 
The lemma follows since $\theta^{-1} \leq |P_{\max}|$. \qed
\end{proof}

\paragraph{Rounding by node distances.} Alternatively,
we can use $x^*$ to obtain a discrete solution. We pick a threshold
$\theta$ uniformly at random from $[0,1]$ (or find the best one), and
choose all nodes $u$ with $x^*(u) \geq \theta$ 
(call this $\Vs_\theta$). This induces the cut $C_\theta = \delta(\Vs_\theta)$.
Since the node labels $x^*$ can also be
considered as distances from $s$, we refer to this rounding methods as
\emph{distance rounding}.

\begin{lemma}\label{lem:roundapprox}
  The worst-case approximation factor for a solution $C_\theta$ obtained
  with distance rounding is
  $\mathbb{E}_\theta[f(C_\theta)] \leq 
  (n-1)\ft(y^*) \leq (n-1)f(C^*)$.
\end{lemma}
\begin{proof}
  
  To upper bound the quantity $\mathbb{E}_\theta[f(C_\theta)]$, we
  partition the set 
  of edges into $(n-1)$ sets $\delta^+(v)$, that is, each set
  corresponds to the outgoing edges of a node $v \in \Vs$. We sort the
  edges in each $\delta^+(v)$ in nondecreasing order by their values $y^*(e)$. Consider
  one specific incidence set $\delta^+(u)$ with edges $e_{u,1}, \ldots,
  e_{u,h}$ and $y^*(e_{u,1}) \leq y^*(e_{u,2}) \leq \ldots \leq
  y^*(e_{u,h})$. Edge $e_{u,i}$ is in the cut if $\theta \in
  [x^*(u),x^*(u)+ y^*(e_{u,i}))$. Therefore, it holds for each node $u$ that
  \begin{align}
    \mathbb{E}_\theta[f(C_\theta \inter \delta^+(u))] &=
    \int_0^1f(C_\theta \inter \delta^+(u)) d\theta\\
    &= \sum_{j=1}^h(y^*(e_{u,j}) - y^*(e_{u,j-1}))f(\{e_{u,j}, \ldots e_{u,h} \})  \\
    &= \ft(y^*({\delta^+(u)})),
  \end{align}
  where we define $y^*(e_{u,0}) = 0$ for convenience, and assume that $f(\emptyset) =0$.
  This implies that
  \begin{align}
    \mathbb{E}_\theta[f(C_\theta)] &\leq \mathbb{E}_\theta[\sum_{v \in
      \Vs}f(C_\theta \inter \delta^+(v))]\\
    &= \sum_{v \in
      \Vs}\ft(y^*({\delta^+(v)})) \leq (n-1)\ft(y^*) \leq (n-1)f(C^*).
  \end{align} \qed
\end{proof}

A more precise approximation factor is $\frac{\sum_v
  \ft(y^*({\delta^+(v)}))}{f(y^*)}$.

\section{Special cases}
\label{sec:easier}

The complexity of \mcoocut{} is not always as bad as 
the worst-case bound in Theorem~\ref{thm:lower}. While it is useful to
consider this most general case (see
Section~\ref{sec:gener-coop-cut}), we next discuss properties of the
submodular cost function and the graph structure that lead to better
approximation factors.  Our discussion is not specific to cooperative
cuts; it is rather a survey of properties that make a number of
submodular optimization problems easier.

\subsection{Separability and sums with bounded support}\label{sec:separability}
An important factor for tractability and approximations is the
separability of the cost function, that is, whether there are
\emph{separators} of $f$ whose structure aligns with the graph. 
\begin{defn}[Separator of $f$]
  A set $S \subseteq \Es$ is called a \emph{separator} of $f:
  2^\Es\to\R$ if for all $B \subseteq \Es$, it holds that $f(B) = f(B
  \inter S) + f(B \setminus S)$. The set of separators of $f$ is
  closed under union and intersection.
\end{defn}

The structure of the separators strongly affects the
complexity 
of \mcoocut{}. 
First and obviously, the extreme case that $f$ is a modular function
(and each $e \in \Es$ is a separator)
can be solved exactly. Second, if the separators of $f$ form a partition
$\Es = 
\bigcup_v E^+_v \union \bigcup_v E^-_v$ that aligns with node neighborhoods such that 
$E_v^+ \subseteq \delta^+(v)$, and $E_v^-\subseteq \delta^-(v)$, then both
$\fap$ and distance rounding solve the problem exactly. No change in the algorithm is needed, i.e., the exact partition need not be known.
In that case, the flow-cut gap is zero, as becomes obvious from the proof of
Lemma~\ref{lem:roundapprox}, since $(\sum_v \ft(y^*_{E_v}))/\ft(y^*) =
1$. 
These separators respect the graph structure and rule out any non-local edge interactions.

\paragraph{Sums of functions with bounded support.}
A generalization of the case of separators are functions that are a sum $f(A) = \sum_{i}f_i(A \inter B_i)$ of functions $f_i$, each of which has bounded support $B_i$. The $B_i$ can be overlapping. In this case, the approximation bounds improve for many of the algorithms in Section~\ref{sec:surrog} that rely on a surrogate function, and for the greedy approximation in Lemma~\ref{lem:ghbound}. Those bounds, summarized in Table~\ref{tab:decomp}, can be shown by approximating each $f_i$ separately by $\hat{f}_i$ with approximation factor $\alpha_i$ that now depends on $|B_i|$, and using $f(A) \leq \max_j \alpha_j \sum_i \fa_i(A)$. This separate approximation is implicit in all those algorithms except the approximation from Section~\ref{subsec:generic}. In those implicit cases, no changes need to be made in the implementation and the partition need not be known. For the generic approximation in Section~\ref{subsec:generic}, one can approximate each $f_i$ explicitly and separately, if the partition is known. Optimizing the resulting sum $\sum_i \fa_i$ or its square is however no longer a minimum cut problem. It admits an FPTAS \citep{nik10,kohliOJ13}.

\begin{table}
  \centering
  \small
  \begin{tabular}{|lc|}
    \hline
    & \\[-8pt]
    generic (\S \ref{subsec:generic}) & $O(\max_i\, \sqrt{|B_i|}\log |B_i|)$ \\[5pt]
    semigradient (\S \ref{subsec:semigrad}) & $\max_i \frac{|C^* \inter B_i|}{(|C^* \inter B_i|-1)(1-\kappa_{f_i}) + 1}$ \\[5pt]
    polymatroidal flow (\S \ref{subsec:pmf}) & $\max_i \min\{\Delta_s \inter B_i, \Delta_t \inter B_i\}$ \\[3pt] 
    deterministic greedy (\S \ref{sec:randgreed}) & $\max_i |\Ch_d \inter B_i|$ \\[2pt] \hline
  \end{tabular}
  \caption{Improved approximation bounds for functions of the form $f(A) = \sum_{i=1}^k f_i(A \inter B_i)$. The bounds are now determined by the largest support $\max_i |B_i|$, but not by $k$.}
  \label{tab:decomp}
\end{table}

For the relaxations, it is not immediately clear that the decomposition always leads to improvements. Consider for example a function $f(A) = f_1(A \inter B_1) + f_2(A \inter B_2)$, where $f_1(B_1) = f_2(B_2)$, $P \triangleq P_{\max} = B_1 \union B_2$ and $|B_1| = |B_2| = |P_{\max}/2|$. Then $\ft(\frac{1}{|P|}\chi_{P}) = \ft(\frac{1}{|B_1|}\chi_{B_1})$. In that case, the proof of Lemma~\ref{lem:ubrelax} may still require $\theta^{-1} = |P_{\max}|$.

\subsection{Symmetry and ``unstructured'' functions}
\label{sec:symmetry}

Going one step further, one may consider sums of that do not necessarily have bounded support but are of a simpler form. 
An important such class are functions $f_i(A)
= g(\sum_{e \in A}w_i(e)) = g(w_i(A))$ for nonnegative weights $w_i(e)$ and a
nondecreasing concave function $g$. We refer to the submodular
functions $g(w(A))$ as \emph{unstructured}, because they only consider
a sum of weights, but otherwise do not make any distinction between
edges (unlike, e.g., graphic matroid rank functions).  One
may classify such functions into a hierarchy, where $\mathcal{F}(k)$
contains all functions $f(A) = \sum_{j=1}^kg_j(w_j(A))$ with at most
$k$ such components. 
The functions $\mathcal{F}(k)$ are special cases of low-rank quasi-concave functions, where $k$ is the rank of the function.

If $k=1$, then it suffices to minimize $w_1(C)$ directly and the problem reduces to \mcut. For $k > 1$, several combinatorial problems admit an FPTAS with running time exponential in $k$ \citep{gr08,ms12}. This holds for cooperative cuts too \citep{kohliOJ13}. A special case for $k=2$ is the \emph{mean-risk} objective $f(A) = w_1(A) +
\sqrt{w_2(A)}$ \citep{nik10}.
\citet{gtw10} show
that these functions can yield better 
bounds in combinatorial multi-agent problems than general polymatroid
rank functions, if each agent has a cost function in
$\mathcal{F}(1)$. 

Even for general, unconstrained submodular 
minimization\footnote{For unconstrained submodular function
  minimization we drop the constraint that the functions $g_j$ are
  nondecreasing.}, the class $\mathcal{F}(k)$ admits specialized
improved optimization algorithms
\citep{kohli2009p,sk10,kol12,jegbs13}. The complexity of those faster
specialized algorithms depends on the rank $k$ as well.
An interesting question arising from the above observations is whether $\mathcal{F}(k)$ contains \emph{all} submodular functions if $k$ is large enough?
The answer is no: even if $k$ is allowed to be
exponentially large in the ground set size, this class is a strict sub-class of all submodular
functions. If the addition of auxiliary variables
is allowed, this class coincides with the class of graph-representable functions
in the sense of \citep{zcj09}: any graph cut function $h: 2^\Vs
\to \R_+$ is in
$\mathcal{F}(|\Es|)$, and any function in $\mathcal{F}(k)$ can be
represented as a graph cut function in an extended auxiliary graph
\citep{jlb11}. However, not all submodular functions can be represented in this way \citep{zcj09}.

The parameter $k$ is a measure of complexity.
If $k$ is not
fixed, then \mcoocut{} is NP-hard; for example, the reduction in
Section~\ref{sec:reduct} uses such functions. Even more, unrestricted
$k$ may induce
large lower bounds, as has been proved for \emph{label
  cost} functions of the form $f(A) = \sum_{j=1}^k w_j\min\{1, |A
\inter B_j|\}$ \citep{zctz09}.  

A subclass of unstructured submodular functions are the aforementioned
\emph{permutation symmetric} submodular functions\footnote{These are distinct
from the other previously-used notion of
symmetric submodular functions 
\citet{q95}
where, for all $A \subseteq \Es$,
$f(A) = f(\Es \setminus A)$.}
that are indifferent to any permutation of the ground set: $f(A)
= f(\sigma(A))$ for all permutations $\sigma$ (possibly within a group).
This symmetry makes submodular optimization problems easier, as shown in Proposition~\ref{prop:symm} and work on submodular maximization  
\citep{vondrak2011} and partitioning problems \citep{ene13}.

\subsection{Symmetry and graph structure}\label{subsec:symm_and_graph}

Proposition~\ref{prop:symm} shows that symmetry and the graph structure can work together to make the cut problem easier, in fact, a submodular minimization problem on graph nodes. Section~\ref{sec:motiv-spec-cases} outlines some examples that come from applications.

\subsection{Curvature}
The \emph{curvature} $\kappa_f \in [0,1]$ of a submodular function $f$ is defined as
\begin{equation}
  \label{eq:1}
  \kappa_f = \max_{e \in \Es}\;\; 1 - \frac{f(e \mid \Es \setminus e)}{f(e)},
\end{equation}
and characterizes the deviation from being a modular function. 
Curvature is known to affect the approximation bounds for submodular maximization \citep{confcorn84,vo08}, and also for submodular minimization problems, approximating and learning submodular functions \citep{iyercurv13}. The lower the curvature, the better the approximation factors.
For \mcoocut{} and many other combinatorial minimization problems with submodular costs, the approximation factor is affected as follows. If $\alpha_n$ is the worst-case factor (e.g., for the semigradient approximation), then the tightened factor is $\frac{\alpha_n}{(\alpha_n-1)(1-\kappa_f)+1}$. The lower bounds can be tightened accordingly.

\subsection{Flow-cut gaps revisited}\label{sec:specialgaps} 
The above properties that facilitate \mcoocut{} reduce the flow-cut gaps in some cases.
The proof of Lemma~\ref{lem:fcgap} illustrates that
the flow-cut gap is intricately linked
to the edge cooperation (non-separability) along paths in the graph. 
Therefore, the separability
described in Section~\ref{sec:separability} affects the flow-cut
gap if it breaks up cooperation along paths: the gap depends only on the longest cooperating path within
any separator of $f$, and this can be much smaller than $n$. 
If, however, an instance of \mcoocut{} is better solvable because the
cost function is a member of $\mathcal{F}(\ell)$ for small constant $\ell$,
then the gap may still be as large as in Lemma~\ref{lem:fcgap}. In fact, the example in Lemma~\ref{lem:fcgap} belongs to $\mathcal{F}(1)$: it is equivalent to the function $f(A) =  \gamma\min\{1,|A|\}$.

Two variants of a final example may serve to better understand the
flow-cut (and integrality) gap. The first has a large gap,
but the
rounding methods still find an optimal solution. The second has a
gap of one, but the rounding methods may return solutions with a large
approximation factor.
Consider a graph with $m$ edges consisting of $m/k$ disjoint
paths of length $k$ each (as in Figure~\ref{fig:lbgraph}), with a
 cost function $f(C) = \max_{e \in C}w(e)$. The edges are
partitioned into a cut $B \subset \Es$ with $|B|=m/k$ and the
remaining edges $\Es \setminus B$. Let $w(e) = \gamma$ for $e \notin
B$ and $w(e) = \beta$ for $e \in B$. 

For the first variant, let $\beta = \gamma$; so that for $k=1$, we
obtain the graph in Lemma~\ref{lem:fcgap}. With $\beta=\gamma$ (for any
$k$), any minimal cut is optimal, and all rounding methods find
an optimal solution. The maximum flow in Problem~\eqref{eq:relflow} is $\nu^* = \gamma/k$
($\gamma / k$ flow on one path or $\gamma / m$ flow on each edge in
$m/k$ paths in parallel). Hence, the flow-cut gap is $\gamma / (\gamma/k) =
k$ despite the optimality of the rounded (and pruned) solutions.

For the second variant, let $\beta = \gamma/k$. The
maximum flow remains $\nu^* = \gamma/k$, and the optimal cut is $B$
with $f(B) = \gamma/k$, so $f(C^*) = \nu^*$. An optimal solution $y^*$
to Program~(\ref{eq:relaxshort}) is the uniform vector $y^* = (\gamma/m)
\mathbf{1}_m$. Despite the zero gap, for such $y^*$ the rounding methods return an arbitrary cut, which
can be by a factor $k$ worse than the optimal solution $B$. 
In contrast, the approximation
algorithms in Sections~\ref{subsec:semigrad}, \ref{subsec:pmf} based on substitute cost functions do return an optimal solution.

\section{Experiments}\label{sec:expt}
We provide a summary of benchmark experiments that compare the
proposed algorithms empirically. 
We use two types of data sets. The first is a collection of average-case submodular cost functions on two types of graph structures, clustered graphs and regular grids. The second consists of a few difficult examples that show the limits of some of the proposed methods.

The task is to find a minimum cooperative cut in an undirected graph\footnote{An undirected graph can easily be turned into a directed one by replacing each edge by two opposing directed ones that have the same cost. A cut will always only include one of those edges}. This problem can be solved directly or via $n-1$ minimum $(s,t)$-cuts. Most of the algorithms solve the $(s,t)$ version. The above approximation bounds still apply, as the minimum cut is the minimum $(s,t)$-cut for at least one pair of source and sink. We observe that, in general, the algorithms perform well, typically much better than their theoretical worst-case bounds. Which algorithm is best depends on the cost function and graph at hand.

\paragraph{Algorithms and baselines.} Apart from the algorithms discussed in this article, we test some baseline heuristics.
First, to test the benefit of the more sophisticated approximations $\fae$ and $\fap$ we define the simple approximation
\begin{equation}
  \label{eq:fadd}
  \hat{f}_{\mathrm{add}}(C) = \sum_{e \in C}f(e).
\end{equation}
The first baseline (MC) simply returns the minimum cut with respect to $\famo$. The second baseline (MB) computes the minimum cut basis $\Cs = \{C_1, \ldots, C_{n-1}\}$ with respect to $\famo$ and then selects $\Ch = \argmin_{C \in \Cs} f(C)$. The minimum cut basis can be computed via a Gomory-Hu tree \citep{bhms07}.
As a last baseline, we apply an algorithm (QU) by \citet{q95} to $h(X) = f(\delta(X))$. This algorithm minimizes symmetric submodular functions in $O(n^3)$ time. However, $h$ not always submodular (e.g., see Props.~\ref{prop:submdoular_iff_modular}, \ref{prop:properties_h},
and~\ref{prop:symm}), and therefore this algorithm cannot provide any approximation guarantees in general. In fact, we will see in Section~\ref{sec:wc} that it can perform arbitrarily poorly.

Of the algorithms described in this article, EA denotes the generic (ellipsoid-based) approximation of Section~\ref{subsec:generic}. The iterative semigradient approximation from Section~\ref{subsec:semigrad} is initialized 
with a random cut basis (RI) or a minimum-weight cut basis (MI).
PF is the approximation via polymatroidal network flows (Section~\ref{subsec:pmf}). These three approaches approximate the cost functions. In addition, we use algorithms that solve relaxations of Problems \eqref{eq:cover} and \eqref{eq:relaxshort}: CR solves the convex relaxation using Matlab's \texttt{fmincon}, and applies Algorithm~\ref{alg:round} for rounding. DB implements the distance rounding by thresholding $x^*$.
Finally, we test the randomized greedy algorithm from Section~\ref{sec:randgreed} with the maximum possible $\beta = \beta_{\max}$ (GM) and an almost maximal $\beta = 0.9\beta_{\max}$ (GA). GH denotes the deterministic greedy heuristic.
All algorithms were implemented in Matlab, with the help of a graph cut toolbox \citep{bagon2006,bk04} and the SFM toolbox \citep{kTool}.

\subsection{Average-case}
The average-case benchmark data has two components: graphs and cost functions. We first describe the graphs, then the functions.

\paragraph{Grid graphs.} The benchmark contains three variants of regular grid graphs of
degree four or six. Type I is a plane grid with horizontal and
vertical edges displayed as solid edges in Figure~\ref{fig:gridgraph}. Type II
is similar, but has additional diagonal edges (dashed in
Figure~\ref{fig:gridgraph}). Type III is a cube 
with plane square grids on four faces (sparing the top and bottom
faces). Different from Type I, the nodes in the top row are connected
to their counterparts on the opposite side of the cube. The
connections of the bottom nodes are analogous.

\paragraph{Clustered graphs.} The clustered graphs consist of a number of
cliques that are connected to each other by few edges, as depicted in
Figure~\ref{fig:clustgraph}.

\begin{figure}
  \centering
  \subfigure[Grids I and II]{\label{fig:gridgraph}
    \includegraphics[width=0.26\textwidth]{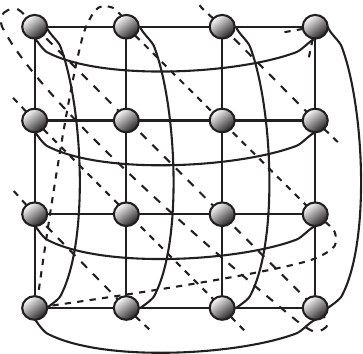} 
  }\hspace{50pt}
  \subfigure[Clustered graph]{\label{fig:clustgraph}
    \includegraphics[width=0.3\textwidth]{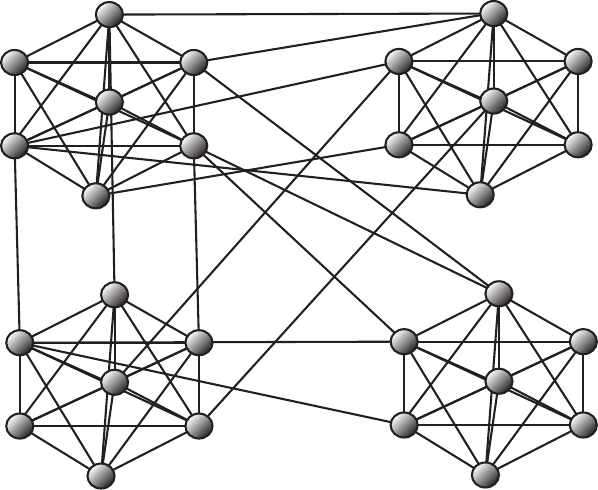} 
  }
  \caption[Examples of benchmark graphs.]{Examples of the test graph structures. The grid (a) was used with and
    without the dashed diagonal edges, and
    also with a variation of the connections in the first and last
    row. The clustered graphs were similar to the example shown in (b).}
  \label{fig:testGraphs}
\end{figure}

\paragraph{Cost functions.}
The benchmark
includes four families of functions. The first group (\emph{Matrix rank I,II}, \emph{Labels I, II}) consists of matroid rank
functions or sums of three such functions. The functions used here
are either based on matrix rank or ranks of partition matroids.
We summarize those functions as \emph{rank-like} costs.

The second group (\emph{Unstructured I, II}) contains two variants of unstructured functions $g(w(C))$, where $g$ is either a logarithm or a square root.
These functions are designed
to favor a certain random optimal cut. The construction ensures that
the minimum cut will not be one that separates out a single node, but
one that cuts several edges.

The third family (\emph{Bestcut I, II}) is constructed to make a cut optimal
that has many edges and that is therefore different from the cut that
uses fewest edges. For such a cut, we expect $\famo$ to yield
relatively poor solutions. 

The fourth set of functions (\emph{Truncated rank}) is inspired by the difficult truncated
functions that can be used to establish lower bounds on approximation
factors. These functions ``hide'' an optimal set, and interactions are
only visible when guessing a large enough part of this hidden set.
The following is a detailed description of all cost functions:

\begin{description}
\item[Matrix rank I.] Each element $e \in
    \Es$ indexes a column in matrix $\mathbf{X} \in \R^{d \times m}$. The cost of $A
    \subseteq \Es$ is the rank of
    the sub-matrix $\mathbf{X}_A$ of the columns indexed by the $e \in
    A$: $f_{\text{mrI}}(A) = \mathrm{rank}(\mathbf{X}_A)$. The matrix $\mathbf{X}$
    is of the form $[\; \mathbf{I'}\;\; \mathbf{R}\;]$, where
    $\mathbf{R} \in \{0,1\}^{d \times (m-d)}$ is a random binary
    matrix with $d = 0.9\sqrt{m}$, and $\mathbf{I}'$ is a
    column-wise permutation of the identity matrix. 
\item[Matrix rank II.] The function $f_{\text{mrII}}(A) = 0.33\sum_{i=1}^3 f^{(i)}_{\text{mrI}}(A)$
    sums up three functions $f^{(i)}_{\text{mrI}}$ of type \emph{matrix rank I} with
    different random matrices $\mathbf{X}$. 
\item[Labels I.] This class consists of functions of the form $f_{\ell\text{I}}(A) = | \bigcup_{e \in A} \ell(e)|$. Each
    element $e$ is assigned a random label $\ell(e)$ from a set of
    $0.8\sqrt{m}$ possible labels. The cost counts the number of
    labels in $A$. 
\item[Labels II.] These functions $f_{\ell\text{II}}(A) = 0.33\sum_{i=1}^3 f^{(i)}_{\ell\text{I}}(A)$ are the sum of three functions of type \emph{labels I} with different
    random labels.
\item[Unstructured I.] These are functions $f_{\text{dpI}}(A) = \log{\sum_{e \in A}
      w(e)}$, where weights $w(e)$ are chosen randomly as follows. Sample
    a set $X \subset V$ with $|X| = 0.4n$, and set $w(e) = 1.001$ for all
    $e \in \delta X$. Then randomly assign some ``heavy'' weights
    in $[n/2, n^2/4]$ to some edges not in $\delta X$, so that each
    node is incident to one or two heavy edges. The remaining edges
    get random (mostly integer) weights between $1.001$ and $n^2/4-n+1$. 
\item[Unstructured II.] These are functions $f_{\text{dpII}}(A) = \sqrt{\sum_{e \in
    A}w(e)}$ with weights assigned as for \emph{unstructured function II}.
\item[Bestcut I.] We randomly pick a connected subset $X^* \subseteq \Vs$ of
    size $0.4n$ and define the cost  
    $f_{\text{bcI}}(A) = \mathbf{1}[|A \inter \delta
    X^*| \geq 1] + \sum_{e \in A \setminus \delta X^*} w(e)$. 
    The edges in $\Es \setminus \delta X^*$ are assigned random
    weights $w(e) \in [1.5,2]$. If there is still a cut $C \neq \delta
    X^*$ with cost one or 
    lower, we correct $w$ by increasing the weight of one $e \in C$ to
    $w(e)=2$. The optimal cut is then $\delta X^*$, but it is usually not the one with fewest edges. 
\item[Bestcut II.] Similar to \emph{bestcut I} ($\delta X^*$ is again
    optimal), but with submodularity on all 
    edges: $\Es$ is partitioned into three sets, $E = (\delta X^*) \union B
    \union C$. Then $f_{\text{bcII}}(A) = \mathbf{1}[|A \inter \delta
    X^*| \geq 1] + \sum_{e \in A \inter(B \union C)} w(e) + \max_{e \in A
      \inter B} w(e) + \max_{e \in A \inter C} w(e)$. The weights of
    two edges in $B$ and two edges in $C$ are set to
    $w(e) \in (2.1,2.2)$.
\item[Truncated rank.] This function is similar to the truncated rank in
    the proof of the lower bound (Theorem~\ref{thm:lower}). Sample a connected $X \subseteq \Vs$
    with $|X| = 0.3|\Vs|$ and set $R = \delta X$. The cost is $f_{\text{tr}}(A)
    = \min\{|A \inter \overline{R}| + \min\{ |A \inter R|, \lambda_1\}, \;
    \lambda_2\}$ for $\lambda_1 = \sqrt{|R|}$ and $\lambda_2 = 2|R|$. Here,
    $R$ is not necessarily the optimal cut.
\end{description}

To
estimate the approximation factor on one problem instance (one graph
and one cost function), we divide by the cost of the best solution
found by any of the algorithms, unless the optimal solution is
known (this is the case for \emph{Bestcut I} and \emph{II}).

\subsubsection{Results}
Figure~\ref{fig:res1} shows average empirical approximation factors
and also the worst observed factors.
The first observation is that all algorithms remain well below their
theoretical approximation bounds\footnote{Most of the bounds proved
  above are absolute, and not asymptotic. The only exception is
  $\fae$. For simplicity, it is here treated as an absolute
  bound.}. That means the theoretical bounds are really 
worst-case results. For several instances we obtain optimal solutions.

\begin{figure}
  \centering
  {\small
  \begin{tabular}{c@{\hspace{-5pt}}c}
     {\bf grid graphs} & {\bf clustered graphs} \\[3pt]
     \multicolumn{2}{c}{rank-like cost functions (average over 61 (left), 80 (right) instances)}\\[-1pt]
     \includegraphics[width=0.49\textwidth]{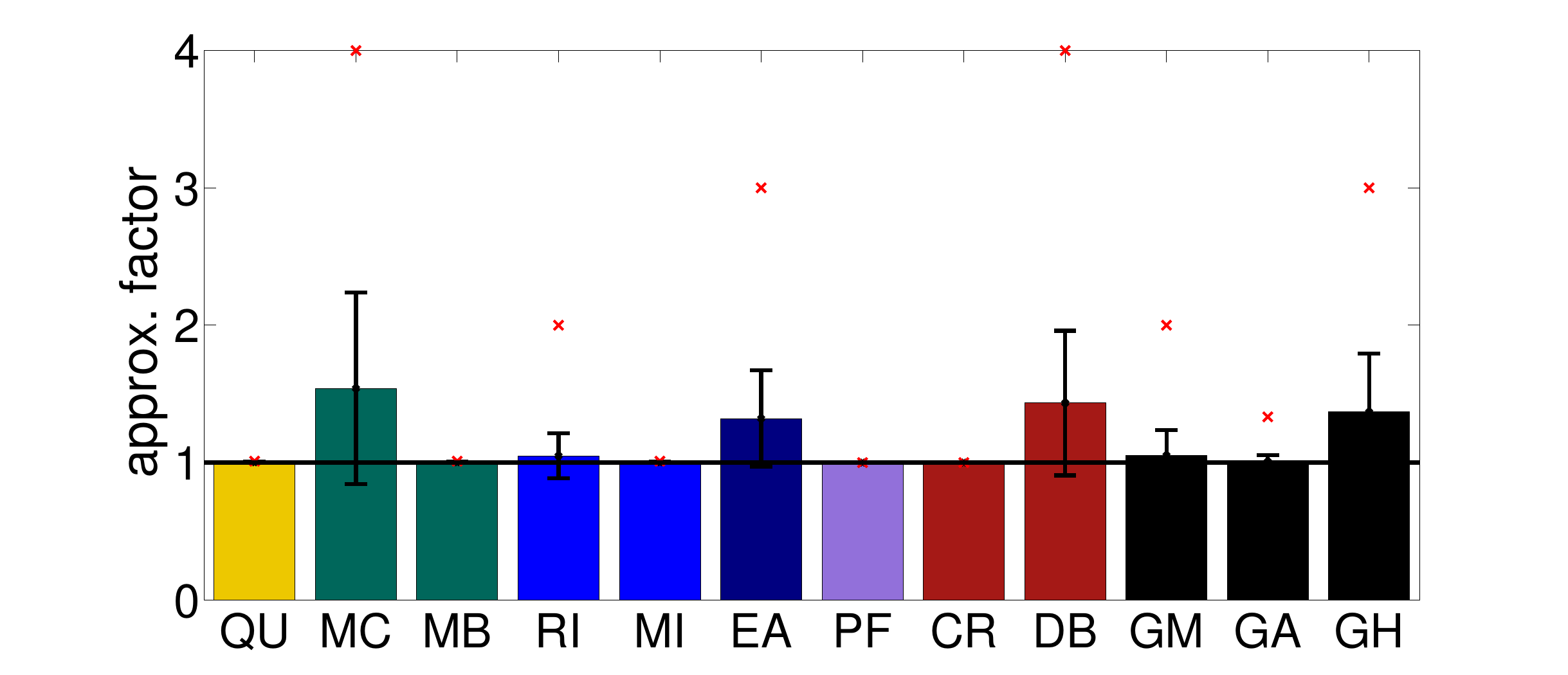} & 
     \includegraphics[width=0.49\textwidth]{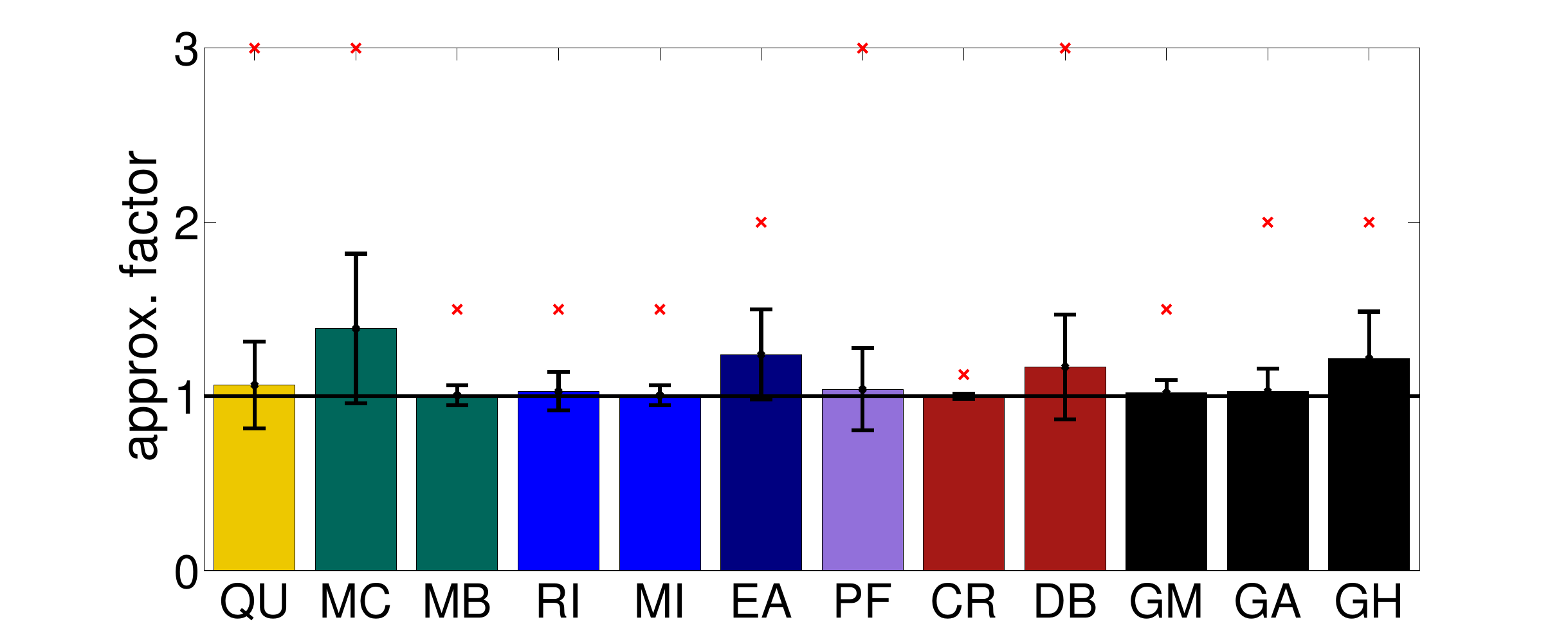} \\[6pt]
     \multicolumn{2}{c}{unstructured functions (average over 30 (left), 40 (right) instances)}\\[-1pt]
     \includegraphics[width=0.49\textwidth]{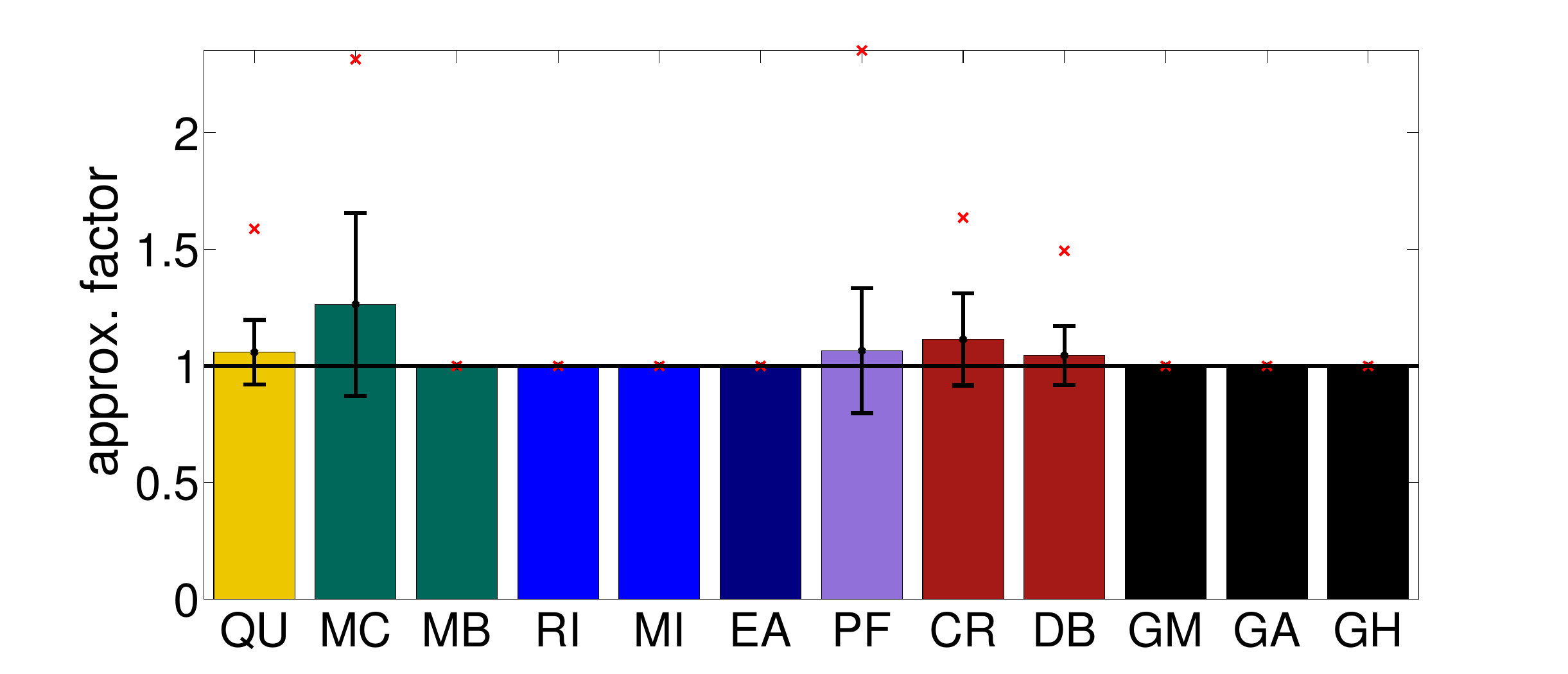} & 
     \includegraphics[width=0.49\textwidth]{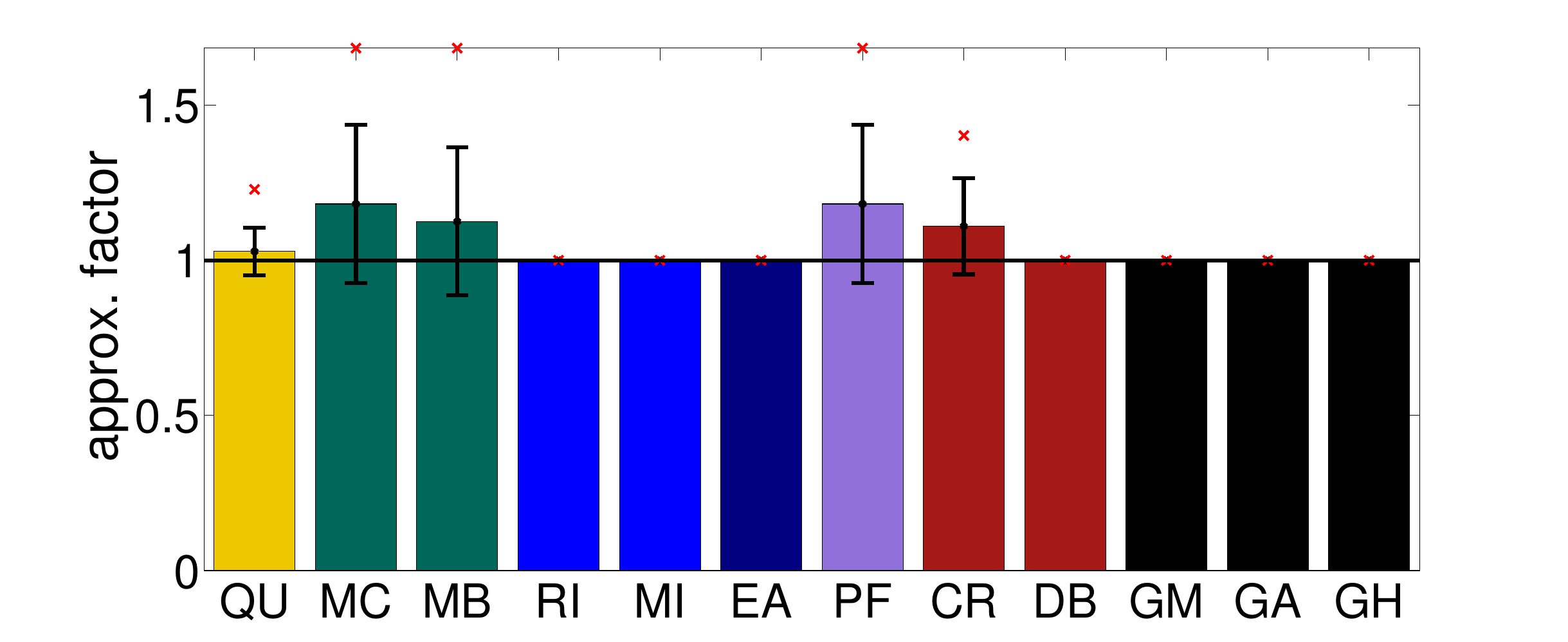} \\[-1pt]
     \multicolumn{2}{c}{bestcut functions (average over 15 (left), 20 (right) instances)}\\[-1pt]
     \includegraphics[width=0.49\textwidth]{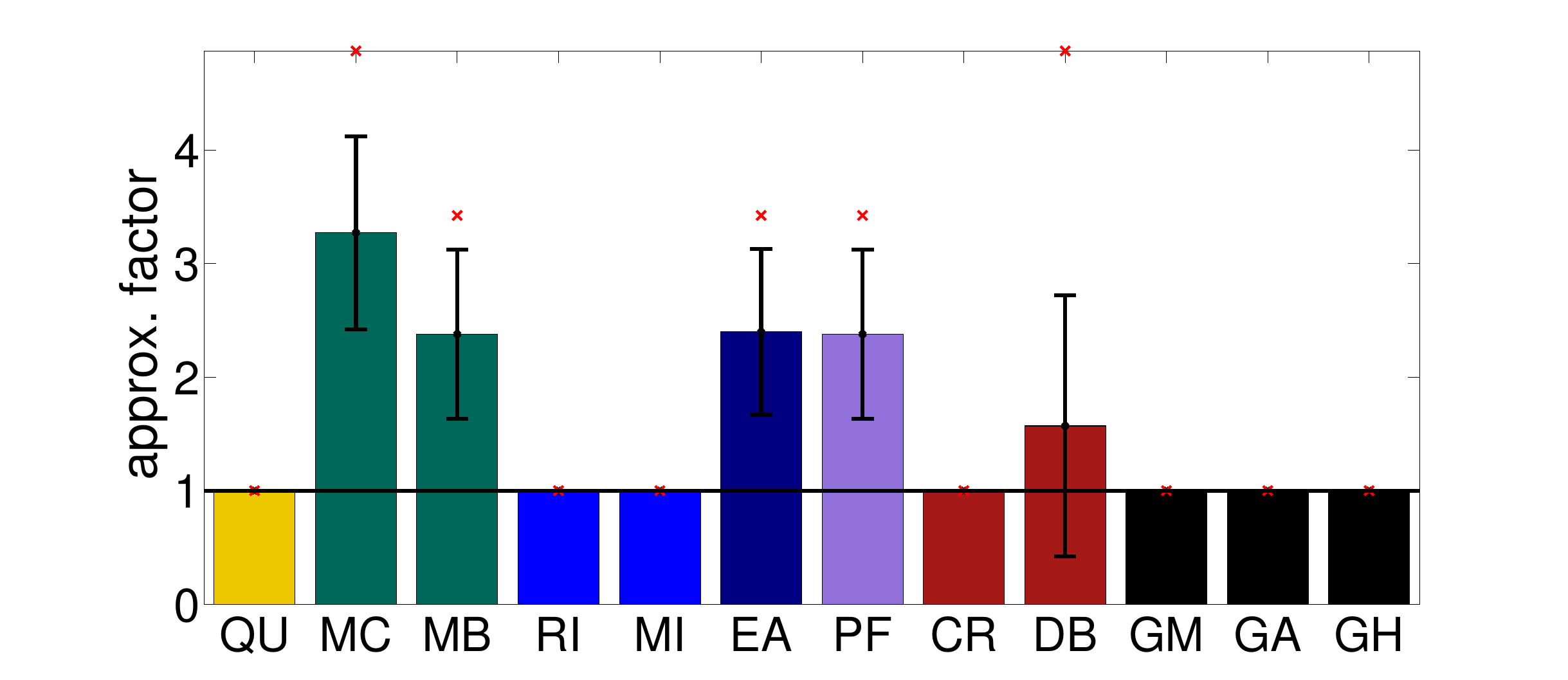} & 
     \includegraphics[width=0.49\textwidth]{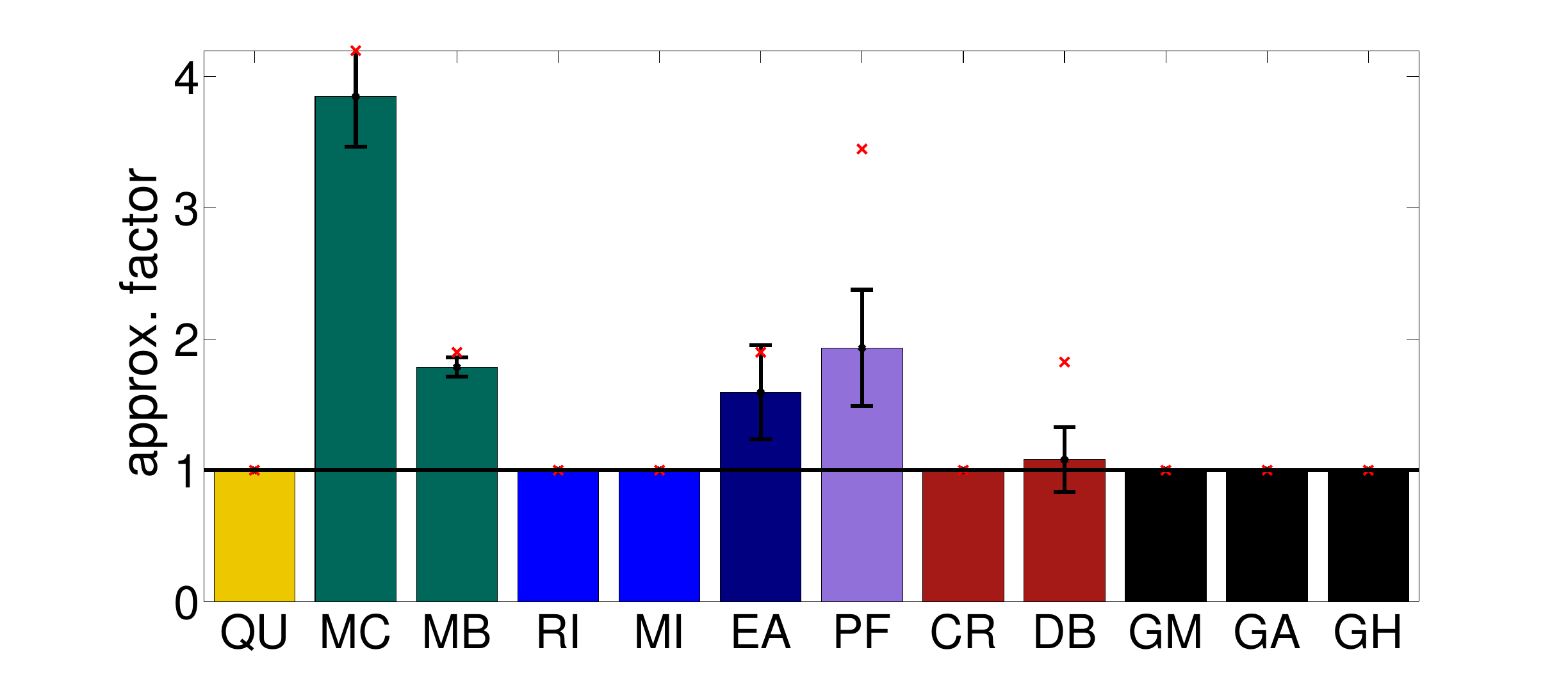} \\[-1pt]
     \includegraphics[width=0.49\textwidth]{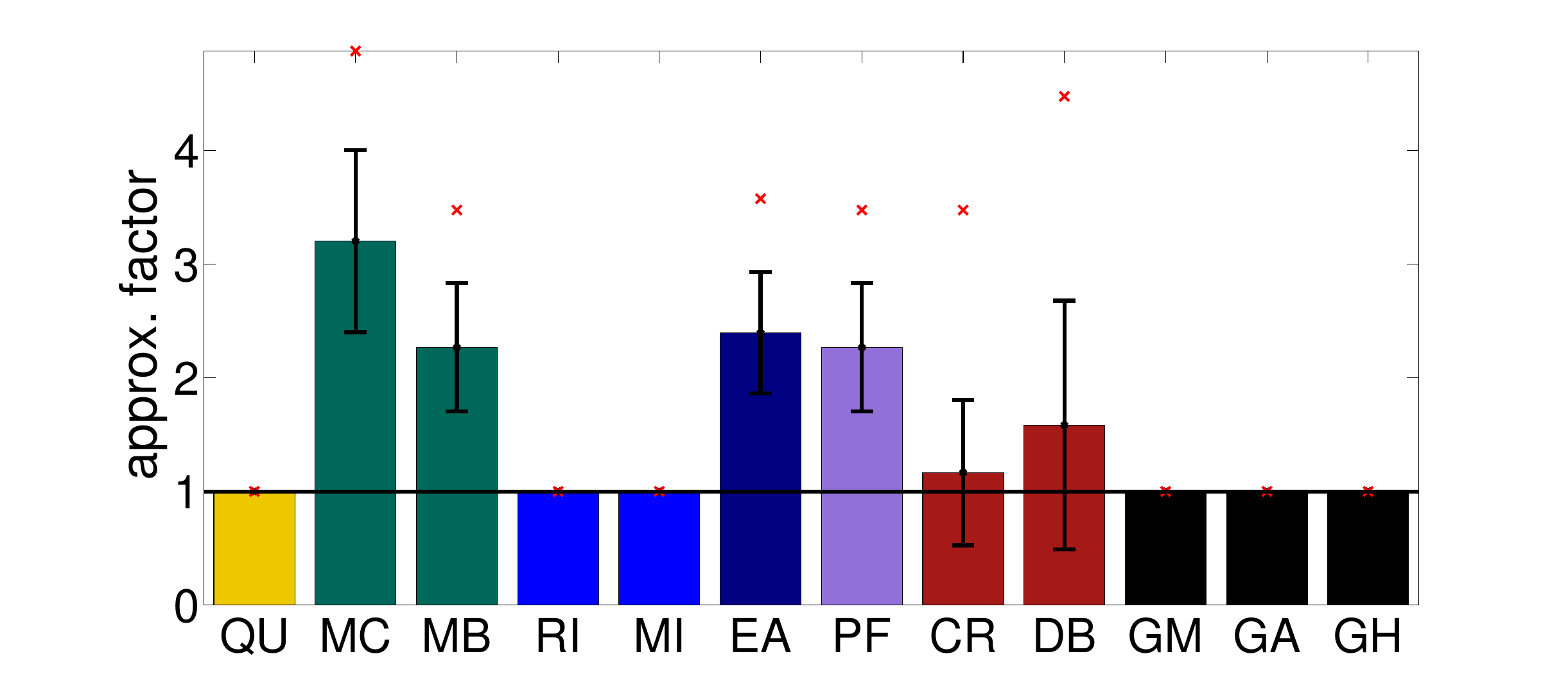} & 
     \includegraphics[width=0.49\textwidth]{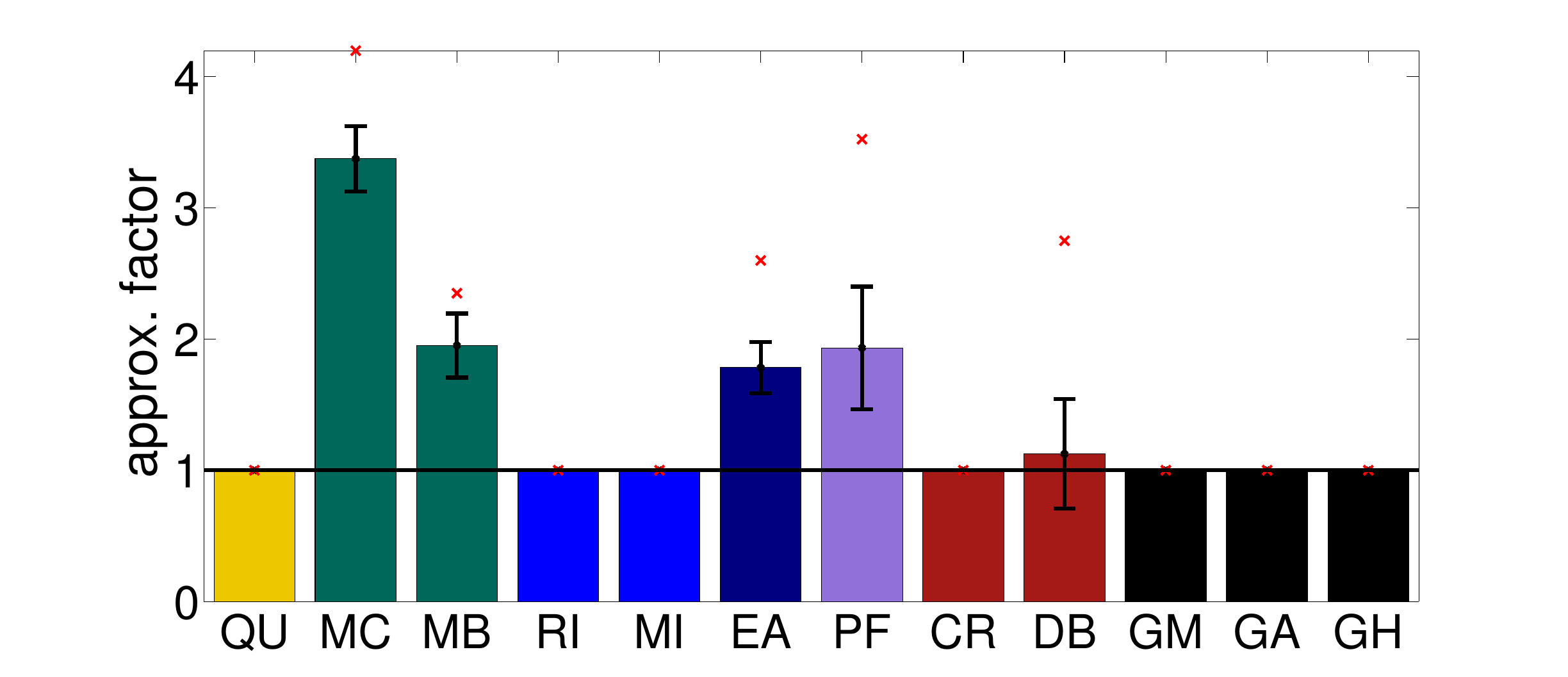} \\[6pt]
     \multicolumn{2}{c}{truncated rank (average over 15 (left), 20 (right) instances)}\\[-1pt]
     \includegraphics[width=0.49\textwidth]{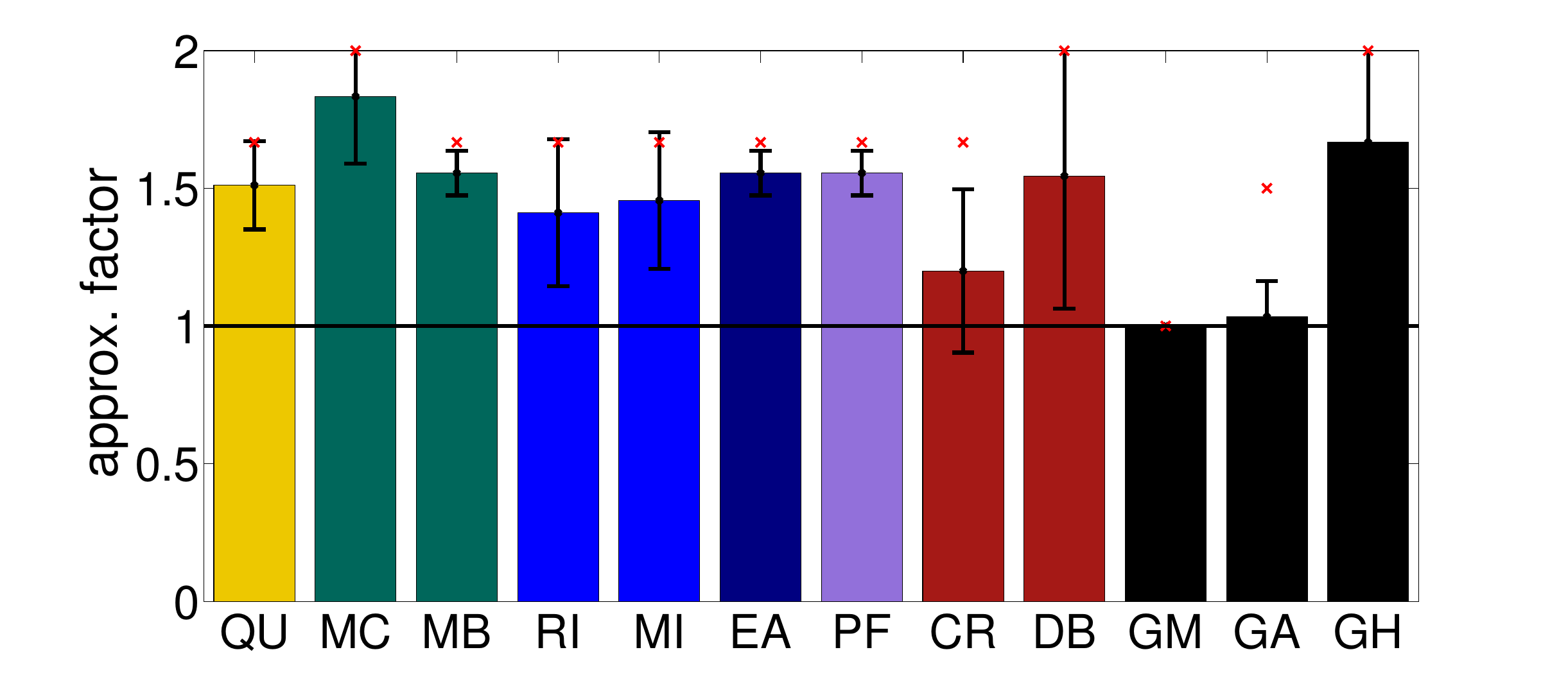} & 
     \includegraphics[width=0.49\textwidth]{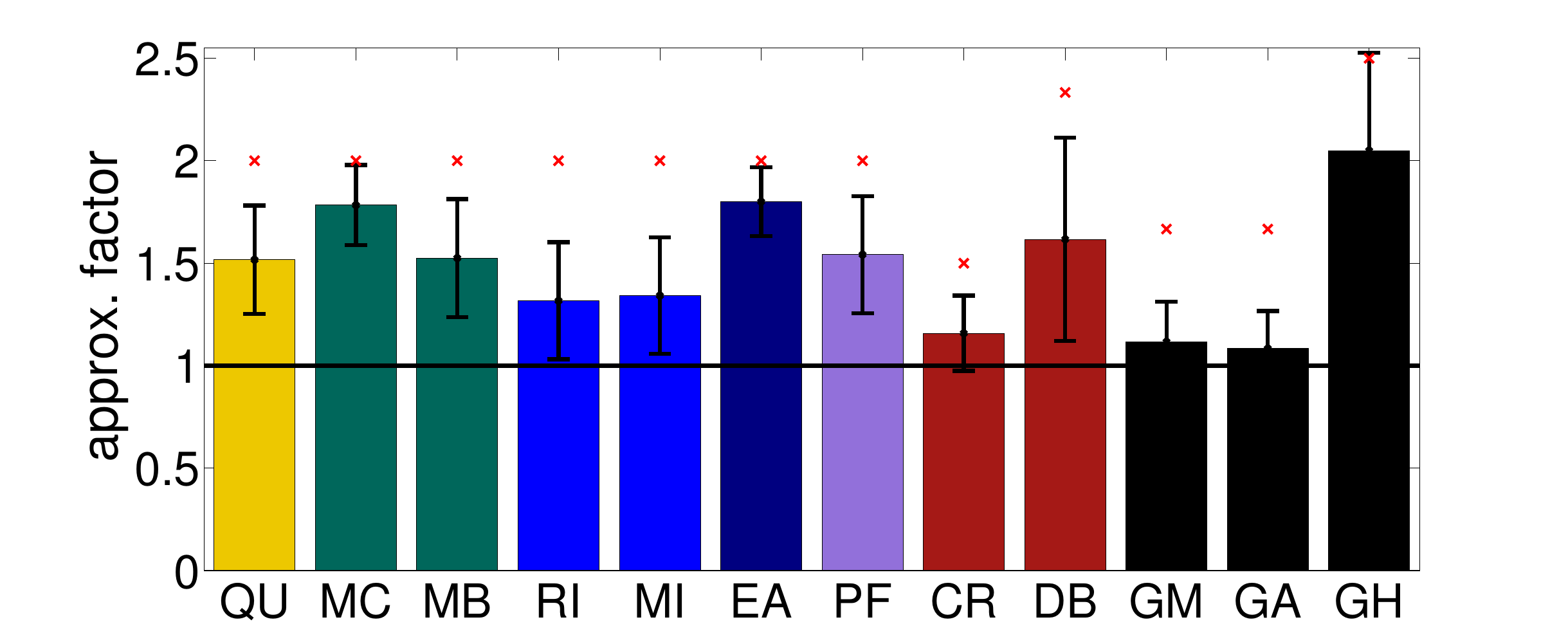} \\
   \end{tabular}
  }
  \caption[Results for the average-case experiments.] {Results for average-case experiments. The bars show the mean empirical approximation
     factors, and red crosses mark the maximum observed empirical approximation
     factor. The left column refers to grid graphs, the right column to clustered graphs. The first three algorithms (bars) are baselines, the next four approximate $f$, the next four solve a relaxation, and the last is the deterministic greedy heuristic.} 
  \label{fig:res1}
\end{figure}

The general performance depends much on the actual problem instance; the truncated rank functions with hidden structure are, as may be expected, the most difficult.
The simple benchmarks relying on $\famo$ perform worse than the more sophisticated algorithms. Queyranne's algorithm performs surprisingly well here.

\subsection{Difficult instances}\label{sec:wc}
Lastly, we show two difficult instances. More examples may be found in~\cite[Ch.~4]{mythesis}. The example demonstrates the drawbacks of using approximations like $\famo$ and Queyranne's algorithm.

Our instance is a graph with $n=10$ modes, shown in Figure~\ref{fig:wc2}. The graph edges are partitioned into $n/2$ sets, indicated by colors. The black set $\Es_k$ makes up the cut with the maximum number of edges. The remaining edge sets are constructed as
\begin{equation}
\Es_i = \big\{(v_i, v_j) \in \Es \;|\; i < j \leq n/2\big\} \union \big\{(v_{n/2+i},
v_j) \in \Es \;|\; n/2 + i < j \leq n\big\} 
\end{equation}
for each $1 \leq i < n/2$. In Figure~\ref{fig:wc2}, set $\Es_1$ is red,
set $\Es_2$ is blue, and so on. The cost function is
\begin{equation}
  f_{a}(A) = \mathbf{1}\big[ |A \inter \Es_k| \geq 1\big] +
  \sum_{i=1}^{n/2-1} b \cdot \mathbf{1}\big[ |A \inter \Es_i| \geq 1\big] + \epsilon |A \inter \Es_k|,
\end{equation}
with $b = n/2$. The function $\mathbf{1}[\cdot]$ denotes
the indicator function. 
The cost of the optimal solution is $f(C^*) = f(\Es_k) = 1 +  \frac{n^2}{4}\epsilon \approx 1$. The second-best solution is the cut $\delta(v_1)$ with cost $f(\delta v_1) = 1 + \frac{n^2}{4}\epsilon  + b \approx 1 + \frac{n}{2} = 6$, i.e., it is by a factor of almost $b = n/2$ worse than the optimal solution. Finally, MC finds the solution $\delta(v_n)$ with $f(\delta v_n) = 1 + \frac{n^2}{4} \epsilon + b(\frac{n}{2}-1) \approx \frac{n^2}{4} = 21$.

Variant (b) uses the cost function
\begin{equation}
  f_{b}(A) = \mathbf{1}[ |A \inter \Es_k| \geq 1] +
  \sum_{i=1}^{n/2-1} b \cdot \mathbf{1}[ |A \inter \Es_i| \geq 1]
\end{equation}
with a large constant $b=n^2=100$. For any $b > n/2$, any solution other than $C^*$ is more than $n^2/4 = |C^*| > n$
times worse than the optimal solution. Hence,
thanks to the upper bounds on their approximation factors, 
all algorithms except for QU find the optimal
solution. The result of the latter
depends on how it selects a minimizer of $f(B \union e) -
f(e)$ in the search for a pendent pair; this quantity often has
several minimizers here. Variant (b) uses a specific adversarial permutation of node labels, for which QU always returns the same solution $\delta v_1$ with cost $b+1$, no
matter how large $b$ is: its solution can become
arbitrarily poor.

\begin{figure}[t]
  \centering
  \begin{tabular}{@{\hspace{-2pt}}p{15pt}@{\hspace{-3pt}}c@{\hspace{-3pt}}p{15pt}@{\hspace{-3pt}}c}
    \multicolumn{4}{c}{\includegraphics[width=0.3\textwidth]{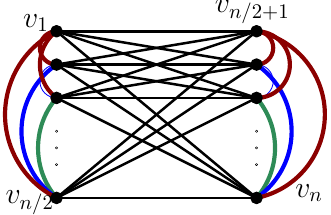}} 
    \\
    (a)& \multirow{2}{*}{\includegraphics[width=0.48\textwidth]{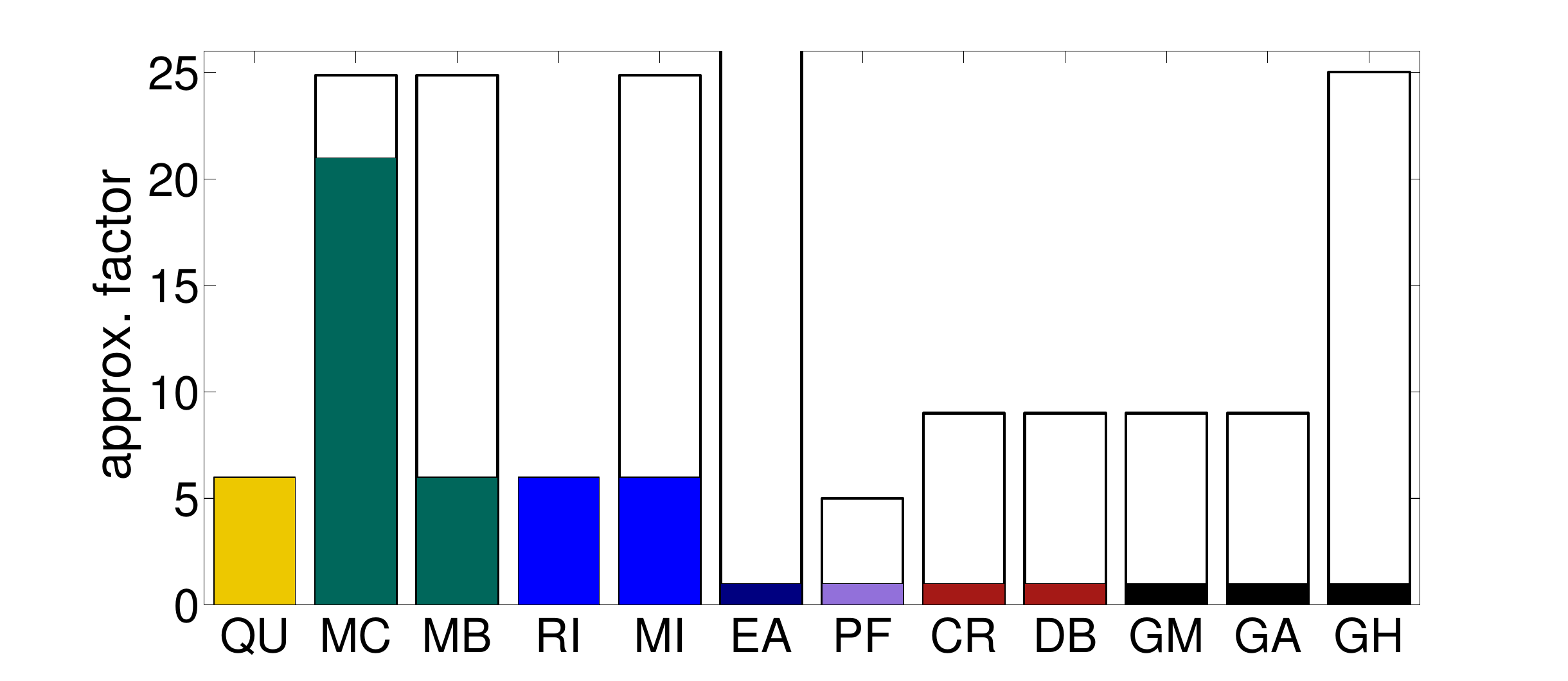}}
     & 
     (b) &
     \multirow{2}{*}{\includegraphics[width=0.48\textwidth]{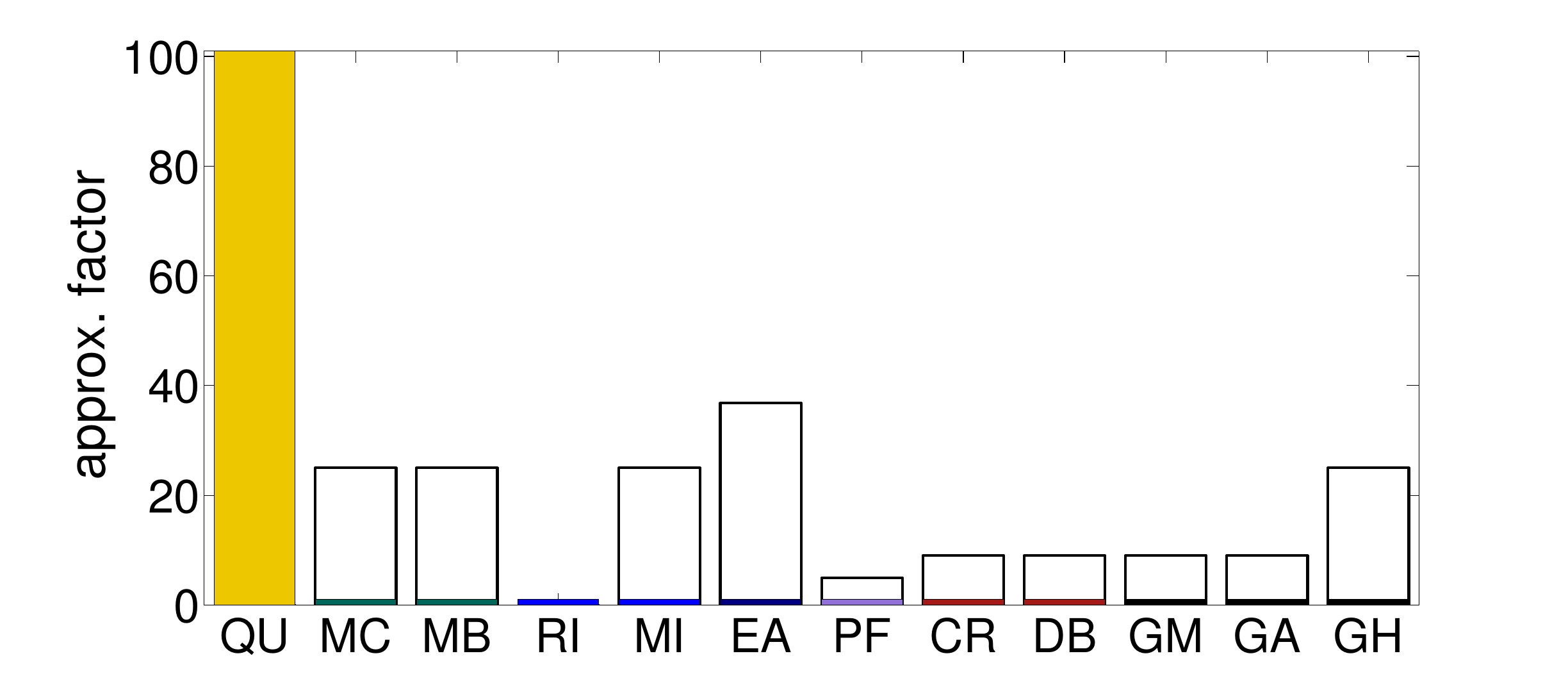}}
     \\ 
      & & & \\
    & & & \\
    & & & \\
    & & & \\[20pt]
  \end{tabular}

  \caption[Difficult graph and empirical approximation factors.]{Difficult instance and empirical approximation factors with $n=10$
    nodes. White bars illustrate theoretical approximation bounds
    where applicable.
    In
    (b), the second-best cut $\delta v_1$ has cost $f_b(\delta v_1) =
    b+1 =101 \gg
    \max\{|C^*|, n, \sqrt{m}\log m\}$.}
  \label{fig:wc2}
\end{figure}

\section{Discussion}
In this work, we have defined and analyzed the \mcoocut{} problem, that is, a
minimum $(s,t)$-cut problem with a submodular cost function on graph
edges. This problem unifies a number of non-additive graph cut problems
in the literature that have arisen in different application areas.

We showed an information-theoretic lower bound of
$\Omega(\sqrt{n})$ for the general \mcoocut{} problem if the function
is given as an oracle, and 
NP-hardness even if the cost function is fully known and polynomially
representable. We propose and compare complementary approximation algorithms  that
either rely on representing the cost function by a simpler function,
or on solving a relaxation of the mathematical program. The latter are
closely tied to the longest path of cooperating edges in the graph,
as is the flow-cut gap. We also show that the flow-cut gap may
be as large as $n-1$, and therefore larger than the best approximation factor possible. 

The lower bound and analysis of the integrality gap use a particular
graph structure, a graph with parallel disjoint paths of equal
length. 
Taken all proposed algorithms
together, all instances of \mcoocut{} on graphs with parallel paths of the same
length can be solved within an approximation bound at most $\sqrt{n}$.
This leaves the question whether there is an instance that makes
\emph{all} approximations worse than $\sqrt{n}$.

Section~\ref{sec:easier} outlined properties of submodular functions
that facilitate submodular minimization under combinatorial
constraints, and also submodular minimization in general. Apart from
separability, we defined the hierarchy of function classes
$\mathcal{F}(k)$. 
The $\mathcal{F}(k)$ are related to
graph-representability and might therefore build a bridge between
recent results about limitations of representing submodular functions
as graph cuts \citep{zcj09} (and, even stricter, the limitations of
polynomial representability) and the results discussed in
Section~\ref{sec:symmetry} that provide improved algorithms whose
complexity depends on $k$.

\subsection{Open problems}
This paper is part of a growing collection of work that studies submodular cost functions in combinatorial optimization problems over cuts, trees, matchings, and so on. Such problems are not only of theoretical interest: they occur in a spectrum of problems in computer vision \citep{jb11,shelhamer14,heng15,taniai15} and machine learning \citep{iyer13,iyer13nips,khalil2014scalable}. In several cases, the functions used do not directly fall into any of the ``easier'' sub-classes (e.g., the entropy cuts outlined in Section~\ref{sec:motiv-spec-cases},
and also see the discussion in Section~\ref{sec:gener-coop-cut}). At the same time, the empirical results in this paper and others \citep{iyer13} suggest that in many cases, the results of approximate algorithms can still be good, even though in the worst case they are not. Section~\ref{sec:easier} outlines beneficial properties. Is there a more precise quantification of the complexity of these problems? Do there exist other properties that lead to better algorithms? One direction that is less explored is the interaction of the graph structure with the cost function.

Specific to this work, cut functions induce a function on nodes. Propositions~\ref{prop:properties_h} and \ref{prop:symm} imply that the node function can be submodular, but in very many cases it is not. Yet, the results for Queyranne's algorithm in Section~\ref{sec:expt} suggest that often the function $h$ may remain close to submodular. This could be the case, for example, if the graph is almost complete and $f$ symmetric, or if the symmetry of $f$ is more restricted. A deeper study of the functions $h$ induced by cooperative cuts could reveal insights about a refined complexity of the problem, and explain the good empirical results. One particular interesting example was the polymatroidal flow case where the function $f$ defined in Eqn.~\eqref{eq:pmf} was not necessarily submodular (see Proposition~\ref{prop:conv}), but where the resulting $h$ (also not necessarily submodular) could be optimized exactly in polynomial time.  This suggests an interesting future direction, namely to fully characterize a class of functions $f$ for which polytime algorithms can be obtained to solve minimum ``interacting cut'' problems (i.e., cut problems where the edges may interact but not necessarily in a purely submodular or supermodular fashion).  The dual polymatroidal flow case shows one instance of interacting cut that can be solved exactly.

Finally, finding optimal bounds and algorithms for related cut problems with submodular edge weights is an open problem. Appendix~\ref{app:sparsestcut} outlines some initial results for cooperative multi-way and sparsest cut.

\bibliographystyle{abbrvnat}
\bibliography{refs}

\appendix

\section{Cooperative multi-cut and sparsest cut}\label{app:sparsestcut}
An extension of \mcoocut{} is the problem of cooperative multi-way cut and
sparsest cut. Using the approximation $\fae$ from Section~\ref{subsec:pmf}, we
can transform any multi-way or sparsest cut problem with a submodular
cost function on edges (instead of a sum of edge weights) into a cut
problem whose cut cost is a convolution of local submodular
functions. The relaxation of this cut problem is dual to the
polymatroidal flow problems considered by \citet{ckrv12}.
Combining their results with ours, we get the following Lemma.
\begin{lemma}
  Let $\alpha$ be the approximation factor for solving a sparsest cut
  / multi-way cut in a polymatroidal network. If we solve a
  cooperative sparsest cut / multi-way cut by first approximating the
  cut cost $f$ by a function $\fae$ and, on this instance, using the
  method with factor $\alpha$, we get an $O(\alpha n)$-approximation
  for cooperative sparsest cut / multi-way cut.
\end{lemma}
Using Theorems~6 and 8 in \citep{ckrv12}, we obtain the following
bounds: 
\begin{corollary}
  There is an $O(n\log k)$ approximation for cooperative sparsest cut
  in undirected graphs that is dual to a maximum multicommodity flow
  problem with $k$ pairs, and an $O(n\log k)$ approximation for
  cooperative multi-way cut.
\end{corollary}
We leave it as an open problem whether these bounds are optimal.

\section{Proof of Proposition~\ref{prop:properties_h}}\label{app:properties}
The first part of Proposition \ref{prop:properties_h} is proven by Figure~\ref{fig:notsubmod}. Here, we show the second part that the function $h(X) = f(\delta^+(X))$ is subadditive if $f$ is nondecreasing and submodular.
Let $X,Y \subseteq \Vs$.
Then it holds that
\begin{align}
  h(X) + h(Y) &= f( \delta^+(X)) + f(\delta^+(Y))\\
  \label{eq:subadd1}
  &\geq f( \delta^+(X) \union \delta^+(Y)) + f( \delta^+(X) \inter \delta^+(Y))\\
  \label{eq:subadd2}
  &\geq f( \delta^+(X) \union \delta^+(Y))\\
  &\geq f( \delta^+(X \union Y))\\
  &= h(X \union Y).
\end{align}
In Inequality~\eqref{eq:subadd1}, we used that $f$ is submodular, and in Inequality~\eqref{eq:subadd2}, we used that $f$ is nonnegative.

\section{Reduction from {\sc Graph Bisection} to \mcoocut{}}
\label{sec:reduct}

In this section, we prove Theorem~\ref{thm:np} via a reduction from
{\sc Graph Bisection}, which is known to be NP-hard \citep{gjs76}.
\begin{defn}[{\sc Graph Bisection}]
  Given a graph $\Gs_B = (\Vs_B,\Es_B)$ with edge weights $w_B: \Es_B \to \mathbb{R}_+$,
  find a partition $V_1 \dot\union V_2 = \Vs_B$ with $|V_1| = |V_2| = |\Vs_B|/2$
  with minimum cut weight $w(\delta(V_1))$. 
\end{defn}

\begin{proof}
  To reduce {\sc Graph Bisection} to \mcoocut{},
  we construct an auxiliary graph $\Gs = (\Vs_B \union \{s,t\}, \Es_B \union \Es_s \union \Es_t)$ that contains an unchanged copy of $\Gs_B$ and two
  additional terminal nodes $s,t$. The submodular weights on the edges
  adjacent to the terminal nodes will express the balance constraint
  $|V_1| = |V_2| = |\Vs_B|/2$. 
  \begin{figure*}[t]
    \subfigure[graph $\Gs$ with $\Es_s$ (blue), 
    $\Es_t$ (red) and $\Gs_B$ (black)]{  \label{fig:redgrapha}
      \begin{minipage}{0.41\textwidth}
        \vspace{-55pt}
        \begin{center}
          \includegraphics[width=0.66\textwidth]{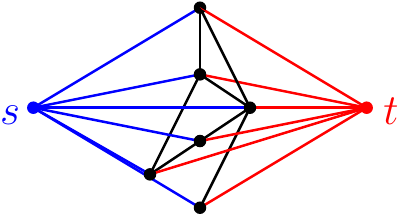}
        \end{center}
      \end{minipage}
    }\hfill 
    \subfigure[graph $\Hs_\sigma$]{ \label{fig:redgraphb}
      \includegraphics[width=0.38\textwidth]{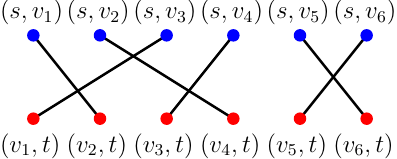}
    }\\ 
    \subfigure[$h_\sigma(\phi(C))=5$ connected components]{
      \includegraphics[width=0.42\textwidth]{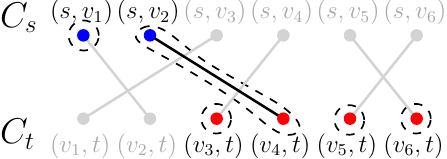}
    }\hfill 
    \subfigure[balanced cut $C$: $h_\sigma(\phi(C))=3$ connected components]{
      \includegraphics[width=0.42\textwidth]{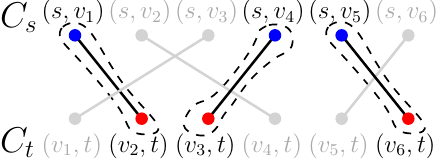}
    }
    \caption[Graph for the reduction and examples for the definition
    of $f_{bal}$ via ranks.]{Graph for the reduction and examples for the definition
      of $f_{bal}$ via ranks $h_\sigma$, with $n_B=6$. In (c),
      $C_s=\{(s,v_1), (s,v_2)\}$ and $C_t=\{(v_3,t), (v_4,t),
      (v_5,t),(v_6,t)\}$; in (d), $C_s=\{(s,v_1), (s,v_4),(s,v_5)\}$
      and $C_t=\{(v_2,t), (v_3,t),(v_6,t)\}$. Connected components
      are indicated by dashed lines.} 
    \label{fig:redgraph}
  \end{figure*}
   In $\Gs$, we retain 
   the modular
   costs $w$ on $\Es_B$ and connect $s,t$ to
   every vertex in $\Gs_B$ with corresponding new edge sets $\Es_s$ and
   $\Es_t$, as illustrated in Figure~\ref{fig:redgrapha}.
   The cost of a cut in $\Gs$ is measured by the
   submodular function
   \begin{equation}
     f(C) = \sum_{e \in C \inter \Es_B}w(e) +  \beta f_{bal}(C \inter (\Es_s
     \union \Es_t)),
   \end{equation}
   where $\beta$ is an appropriately large constant, and $f_{bal}$ will be defined later. 
   Obviously, any minimal $(s,t)$-cut $C$ must include $n_B=|\Vs_B|$ edges
   from $\Es_s \union \Es_t$, and partitions $\Vs_B$. 
   Moreover, the cardinality of $C_s = C \inter
   \Es_s$ is the number of nodes in $\Vs_B$ assigned to $t$. Hence,
   in an equipartition, $|C_s| = |C_t| = n_B/2$, where $C_t = C \inter \Es_t$.

   It remains to define $f_{bal}$ as a nondecreasing submodular function that implements the equipartition constraint. The function will be an expectation over matroid rank functions  $h_\sigma$.
   Let $\Hs_\sigma = (\Es_s, \Es_t, \Fs_{\sigma})$ be a bipartite graph with \emph{nodes} $\Es_s \union \Es_t$ whose edges $\Fs_{\sigma}$ form a derangement 
   between $\Es_s$ and $\Es_t$, as
   illustrated in Fig.\ \ref{fig:redgraphb}.

   Let $\phi(C_s \union C_t)$ be the image of $C_s \union C_t$ in the set of nodes of $\Hs_\sigma$. The function
   $h_{\sigma}: 2^{\phi(\Es_s \union \Es_t)} \to \mathbb{N}_0$ counts the number of connected
   components in the subgraph induced by the nodes $\phi(C_s \union
   C_t)$, and is the rank of a partition matroid. Figure~\ref{fig:redgraph} shows some examples. 

   Let $\mathfrak{S}$ be the set of all derangements $\sigma$ of $n_B$
   elements, i.e., all possible edge configurations in
   $\Hs_\sigma$. We define $f_{bal}$ to be the expectation (under uniform draws of $\sigma$)
   \begin{equation}\label{eq:fb2}
     f_{bal}(C) = \mathbb{E}_\sigma[h_\sigma(\phi(C))] =
     |\mathfrak{S}|^{-1}\sum\nolimits_{\sigma \in \mathfrak{S}} h_\sigma(\phi(C)).
   \end{equation}
   For a fixed derangement $\sigma'$ and a fixed size $|C_s
   \union C_t|=n_B$, the value $h_{\sigma'}(C_s \union C_t)$ is minimal if $\sigma'(C_s) = C_t$ and  $|C_s| = |C_t|$.
   For a fixed $\sigma$, the rank $h_{\sigma}(C)$ is $|\phi(C_s \union C_t)| = |C_s|+|C_t|$ minus the matching nodes. Denoting $(s,v_i)$ in $\Hs_{\sigma}$ by $x_i$ and $(v_i,t)$ by
   $y_i$, the rank is
   \begin{align}
     h_{\sigma}(\phi(C_s) \union \phi(C_t))
     &= |C_s| + |C_t|
     \label{eqn:rankfunction}
     - {\Big| \big\{ (x_i, y_{\sigma(i)}) \big\}_{i=1}^n\;
       \inter\; (\phi(C_s) \times \phi(C_t)) \Big|}.
   \end{align}
   Hence, the sum in \eqref{eq:fb2} becomes
   \begin{align}
     \sum_{\sigma \in \mathfrak{S}} h_\sigma(C)
     &= |\mathfrak{S}|\big(|C_s| +
     |C_t|\big) 
     -\sum_{\sigma \in \mathfrak{S}} \big| \big\{ (x_i,
     y_{\sigma(i)}) \big\}_{i=1}^n\, 
     \inter\, (\phi(C_s) \times \phi(C_t)) \big| \\
     \label{eq:coincRewrite}
     &= |\mathfrak{S}|\big(|C_s| +
     |C_t|\big) 
     - \sum_{x_i \in \phi(C_s)} \sum_{\sigma \in
       \mathfrak{S}} \big| (x_i, y_{\sigma(i)}) \inter  (\{x_i\} \times
     \phi(C_t)) \big| 
   \end{align}
   To compute the sum over $\sigma$ in the second term, let $C_{s\inter t} \triangleq \{ (s, v)\, |\, \{(s,v),(v, t)\} \subseteq C\}$ be the set of $s$-edges whose counterpart on the $t$
side is also contained in $C$. Let further $D'(n_B-1)$
denote the number of permutations of $n_B-1$ elements (pair
$(x_i,y_k)$, i.e., $\sigma(i)=k$, is fixed), where one specific
element $x_k$ can be mapped to
any other of the $n_B-1$ elements, and the remaining elements must not
be mapped to their counterparts ($\sigma(j) \neq j$). Then there are $D'(n_B-1)$ derangements $\sigma$ realizing a specific mapping $\sigma(i)=k$.
Denoting the number of derangements of $n$ elements by $D(n)$, the sum above becomes
\begin{align}
  D(n_B) f_{bal}(C)
  &= (|C_s| + |C_t|)D(n_B) \\
  \nonumber
  &\quad\quad 
  - \sum_{x_i \in C_s \setminus C_{s
      \inter t}} \sum_{y_k \in C_t} D'(n_B-1) 
  - \sum_{x_i \in C_{s
      \inter t}} \sum_{y_k \in C_t, k\neq i} D'(n_B-1)\\[4pt]
  &= (|C_s| + |C_t|)D(n_B)\\[2pt]
  \nonumber
  &\quad\quad- \big(|C_s| - |C_{s \inter t}|\big)|C_t|D'(n_B-1) 
  - |C_{s \inter t}| (|C_t|-1) D'(n_B-1)\\[4pt]
  &= (|C_s| + |C_t|)D(n_B)  
  - (|C_s||C_t| - |C_{s \inter t}|)
  D'(n_B-1),
\end{align}
with $D(n) = |\mathfrak{S}| = n! \sum_{k=0}^n (-1)^k / k!$ \citep{st97}, and $D'(n-1) =
\sum_{k=0}^{n-1} (n-2)!(n-1-k)!(-1)^k$ by Proposition \ref{prop:derange} below.
   
Given that $|C_s| + |C_t|$ must cut at
  least $n_B$ edges and that $f_{bal}$ is increasing, $f_{bal}$ is minimized if
  $|C_s| = |C_t| = n_B/2$. 
  As a result, if $\beta$ is large enough such that $f_{bal}$ dominates the cost, then a minimum cooperative cut in $\Gs$ bisects the $\Gs_B$ subgraph of $\Gs$ optimally. \qed
\end{proof}

\begin{proposition}\label{prop:derange}
  Let $D'(n)$ be the number of permutations of $n$ elements where for one fixed element $i'$ we allow $\sigma(i') \in \{1, \ldots, n\}$, but $\sigma(i) \neq i$ for all $i \neq i'$. Then
  $D'(n) = \sum_{k=0}^{n} \frac{(n-1)!}{k!}(n-k)! (-1)^k$.
\end{proposition}
\begin{proof}
$D'(n)$ can be derived by the method of the forbidden board
\citep[pp.~71-73]{st97}. Let, without loss of generality, $i'=n$, so the
forbidden board is $B = \{(1,1),
(2,2), \ldots, (n-1,n-1)\}$. Let $N_j$ be the number of permutations
$\sigma$ for which $\big|\{(i,\sigma(i)\}_{i=1}^n \inter B\big| = j$, and
let $r_k$ be the
number of $k$-subsets of $B$ such that no two elements have a
coordinate in common. Here, $r_k~=~{n-1 \choose k}$. Then  $D'(n) = N_0 = N_n(0)$ for 
\begin{align}
  N_n(x) &= \sum_j N_j x^j = \sum_{k=0}^n r_k (n-k)!(x-1)^k 
  = \sum_{k=0}^{n} \frac{(n-1)!}{k!} (n-k) (x-1)^k,
\end{align}
and hence $D'(n) = \sum_{k=0}^{n} \frac{(n-1)!}{k!}
(n-k)! (-1)^k$. 
\end{proof}
\section{Convolutions of submodular functions are not always submodular}
\label{app:convolution}

The non-submodularity of convolutions was mentioned already in \citep{lo83}. For completeness, we show an
explicit example that illustrates that non-submodularity also holds for the special case of polymatroidal flows.
\begin{proposition}
\label{prop:conv}
  The convolution of two submodular functions $(f \ast g)(A) = \min_{B \subseteq A}f(B) + g(A\setminus B)$ is not in general submodular. In particular, this also holds for the cut cost functions occurring in the dual problems of polymatroidal maximum flows.
\end{proposition}

\begin{figure}
  \centering
  \begin{minipage}{0.45\linewidth}
    \begin{center}
      \includegraphics[width=0.55\textwidth]{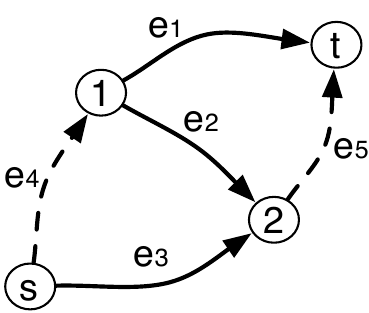}
    \end{center}
  \end{minipage}
  \begin{minipage}{0.45\linewidth}
    Let $f(A) = \max_{e \in A}w(e)$ and\\ $w(e_1) = w(e_2) = a$,\\ $w(e_3) = b$,\\ $w(e_4) = w(e_5) = \epsilon$.
  \end{minipage}
  \caption{Example showing that the convolution of submodular functions is not always submodular, e.g., for $a = 1.5$, $b=2$ and $\epsilon = 0.001$.}
  \label{fig:conv}
\end{figure}

To show Proposition~\ref{prop:conv}, consider the graph in Figure~\ref{prop:conv} with a submodular edge cost function $f(A) = \max_{e \in A}w(e)$. The two submodular functions that are convolved in the corresponding polymatroidal flow are the decompositions
\begin{align}
  \mathrm{cap}^{\text{out}}(A) &= \sum_{v \in \Vs} f(A \inter \delta^+(v))\\
  \mathrm{cap}^{\text{in}}(A) &= \sum_{v \in \Vs} f(A \inter \delta^-(v)).
\end{align}
Both $\mathrm{cap}^{\text{out}}$ and $\mathrm{cap}^{\text{in}}$ are submodular functions from $2^\Es$ to $\mathbb{R}_+$. Their convolution $h$ is
\begin{align}
  h(A) = (\mathrm{cap}^{\text{out}} \ast \mathrm{cap}^{\text{in}})(A) = \min_{B \subseteq A}\mathrm{cap}^{\text{out}}(B) + \mathrm{cap}^{\text{in}}(A \setminus B) = \fap(A).
\end{align}
For $h$ to be submodular, it must satisfy the condition of diminishing marginal costs, i.e., for any $e$ and $A \subseteq B \subseteq \Es \setminus e$, it must hold that $h(e \mid A) \geq h(e \mid B)$.
Now, let $A = \{e_2\}$ and $B = \{e_1, e_2\}$. The convolution here means to pair $e_3$ either with $e_1$ or $e_2$. Then, if $a < b$,
\begin{align}
  h(e_3 \mid A) &= \min\{a+b, b\} - a = b-a\\
  h(e_3 \mid B) &= a + b - \min\{a+a,a\} = b.
\end{align}
Hence, $h(e_3 \mid A) < h(e_3 \mid B)$, disproving submodularity of $h$.

\section{Cooperative Cuts and Polymatroidal Networks}\label{app:dualPMF}
We next prove Lemma~\ref{lem:dualpoly} that relates the
approximation $\fap$ to maxflow problems in polymatroidal networks.

\begin{proof} \emph{(Lemma~\ref{lem:dualpoly})}
  The first step is the dual of the polymatroidal flow. Let
$\mathrm{cap}^{\text{in}}: 2^\Es \to \mathbb{R}_+$ be the joint incoming
capacity, $\mathrm{cap}^{\text{in}}(C) = \sum_{v \in
  V}\mathrm{cap}_v^{\text{in}}(C \inter \delta^-v)$, and let equivalently
$\mathrm{cap}^{\text{out}}$ be the joint outgoing capacity. The dual of
the polymatroidal maximum flow is a minimum cut problem whose cost is a
convolution of edge capacities \citep{lo83}:
\begin{equation}
  \mathrm{cap}(C) =
  (\mathrm{cap}^{\text{in}} \ast \mathrm{cap}^{\text{out}})(C)
  \;\triangleq\; \min_{A \subseteq C} \bigl[ \mathrm{cap}^{\text{in}}(A) +
  \mathrm{cap}^{\text{out}}(C \setminus A) \bigr].
\end{equation}  
We will relate this dual to the approximation $\fap$. Given a minimal
$(s,t)$-cut $C$, let $\Pi(C)$ be a partition of $C$,
and $C_v^{\text{in}} = C^\Pi_v \inter \delta^-_v$ and $C_v^{\text{out}} = C^\Pi_v \inter \delta^+_v$. The cut $C$ partitions the nodes into two sets $\Vs_s$ containing $s$ and $\Vs_t$ containing $t$.
Since $C$ is a minimal directed cut, it contains only edges from the
$s$ side $\Vs_s$ to the $t$ side $\Vs_t$ of the graph. In consequence,
$C_v^{\text{in}} = \emptyset$ if $v$ is on the $s$ side, and
$C_v^{\text{out}} = \emptyset$ otherwise. Hence, $C^{\text{in}}_v
\union C^{\text{out}}_v$ is equal to either $C^{\text{in}}_v$ or
$C^{\text{out}}_v$, and since $f(\emptyset)=0$, it holds that
$f(C^{\text{in}}_v \union C^{\text{out}}_v) = f(C^{\text{in}}_v) +
f(C^{\text{out}}_v)$. 
Then, starting with the definition of $\fap$,
\begin{align}
  \label{eq:e1}
  \fap(C) &= \min_{\Pi(C) \in \mathcal{P}_C}\nlsum_{v \in \Vs}f(C^{\Pi}_v)\\
  &= \min_{\Pi(C)  \in \mathcal{P}_C} \nlsum_{v \in
    \Vs}f(C^{\text{in}}_v \union C^{\text{out}}_v)\\ 
  \label{eq:e2}
  &= \min_{\Pi(C)  \in \mathcal{P}_C}
  \nlsum_{v \in \Vs}\bigl[f(C^{\text{in}}_v) + f(C^{\text{out}}_v)\bigr]\\
  &= \min_{\Pi(C)  \in \mathcal{P}_C} 
  \nlsum_{v \in \Vs}\bigl[\mathrm{cap}_v^{\text{in}}(C^{\text{in}}_v) + \mathrm{cap}_v^{\text{out}}(C^{\text{out}}_v)\bigr]\\
  \label{eq:e3}
  &= \min_{C^{\text{in}},C^{\text{out}}} \bigl[\mathrm{cap}^{\text{in}}(C^{\text{in}}) + \mathrm{cap}^{\text{out}}(C^{\text{out}})\bigr]\\
  \label{eq:e4}
  &= \min_{C^{\text{in}} \subseteq C} \bigl[\mathrm{cap}^{\text{in}}(C^{\text{in}}) + \mathrm{cap}^{\text{out}}(C \setminus C^{\text{in}})\bigr]\\
  \label{eq:e5}
  &= (\mathrm{cap}^{\text{in}} \ast \mathrm{cap}^{\text{out}})(C).
\end{align}
The minimum in Equation~(\ref{eq:e2}) is taken over all feasible partitions $\Pi(C)$ and
their resulting intersections with the sets $\delta^+v, \delta^-v$. Then we use the
notation $C^{\text{in}} = \bigcup_{v \in \Vs}C_v^{\text{in}}$ for all
edges assigned to their head nodes, and $C^{\text{out}} = \bigcup_{v
  \in \Vs}C_v^{\text{out}}$. The minima in Equations~(\ref{eq:e3}) and
(\ref{eq:e4}) are again taken over all partitions in $\mathcal{P}_C$. The
final equality follows from the above definition of a convolution of
submodular functions. \qed
\end{proof}

\end{document}